\documentclass[11pt]{article}
\usepackage[margin=1in]{geometry}
\usepackage{amsfonts}
\usepackage{amsmath,mathtools}
\usepackage{amssymb}
\usepackage{amsthm,xcolor,biblatex}

\newtheorem{lem}{Lemma}
\newtheorem{thm}[lem]{Theorem}
\newtheorem{conj}[lem]{Conjecture}
\newtheorem{prop}[lem]{Proposition}

\usepackage{enumitem}
\usepackage{mathtools}
\usepackage{microtype}
\usepackage{graphicx}
\usepackage{tabu}
\usepackage{tikz-cd}
\usepackage{titlesec,hyperref}
\hypersetup{
    colorlinks=true,
    linkcolor=blue,
    filecolor=magenta, 
    citecolor=red,
    urlcolor=cyan,
    pdftitle={Cryptographic Hash Functions from Markoff Triples},
    pdfpagemode=FullScreen,
}

\newcommand{\A}{\mathbb{A}}
\newcommand{\C}{\mathbb{C}}
\newcommand{\T}{\mathbb{T}}
\newcommand{\SL}{\mathrm{SL}}

\newcommand{\F}{\mathbb{F}}
\newcommand{\Hom}{\operatorname{Hom}}

\newcommand{\ip}[1]{\langle#1\rangle}

\newcommand{\leg}[2]{\of{\frac{#1}{#2}}}
\newcommand{\Mod}[1]{\ (\mathrm{mod}\ #1)}
\newcommand{\of}[1]{\left(#1\right)}
\newcommand{\ov}[1]{\overline{#1}}
\newcommand{\Q}{\mathbb{Q}}

\newcommand{\rot}{\operatorname{rot}}

\newcommand{\Z}{\mathbb{Z}}

\usepackage{authblk}

\title{A Cryptographic Hash Function from Markoff Triples}
\author[1]{Elena Fuchs}
\author[2]{Kristin Lauter}
\author[1]{Matthew Litman}
\author[1]{Austin Tran}

\affil[1]{{\small Department of Mathematics, University of California, Davis}}
\affil[2]{{\small Facebook AI Research, Seattle, WA}}

\begin{document}
\maketitle
\begin{abstract}
    Cryptographic hash functions from expander graphs were proposed by Charles, Goren, and Lauter in \cite{CGL} based on the hardness of finding paths in the graph.  In this paper, we propose a new candidate for a hash function based on the hardness of finding paths in the graph of Markoff triples modulo $p$. These graphs have been studied extensively in number theory and various other fields, and yet finding paths in the graphs remains difficult. We discuss the hardness of finding paths between points, based on the structure of the Markoff graphs. We investigate several possible avenues for attack and estimate their running time to be greater than $O(p)$. In particular, we analyze a recent groundbreaking proof in \cite{BGS1} that such graphs are connected and discuss how this proof gives an algorithm for finding paths.
\end{abstract}

\setcounter{tocdepth}{2}
\tableofcontents

\section{Introduction}
In this work, we introduce a proposal for a hash function based on the hardness of finding paths in the graph of Markoff triples, which we will define.  The idea of using the hardness of path-finding in graphs to define cryptosystems was introduced at the NIST Hash function workshop in 2005~\cite{CGL}. The paper~\cite{CGL} proposed two different candidate families of Ramanujan graphs: 1) LPS Cayley graphs, and 2) Supersingular Isogeny Graphs.  The LPS-based hash function was attacked in two subsequent papers, which presented efficient algorithms to find collisions \cite{TZ}, and preimages \cite{PLQ}. Path-finding in Supersingular Isogeny Graphs remains a hard problem in cryptography so far, and is the basis for the SIDH Key Exchange Protocol~\cite{JFP,CFLMP} in the third round of the NIST PQC competition.

In this paper, we propose to use graphs based on solutions to Markoff's equation to construct a new cryptographic hash function, and discuss why it appears that these graphs may be good candidates. Our main focus will  be to evaluate the path-finding algorithm that can be extracted from the proof of Bourgain-Gamburd-Sarnak in \cite{BGS1} that these graphs are connected in most cases, as this is currently the only certain way to find paths in these graphs in general. We will hence go into the details of the proof in \cite{BGS1} and how it yields an algorithm to find paths, as well as explore some other potential attacks in Section~\ref{Attack2} and \ref{otherattacks}.

\subsection{Markoff tree and graph}
Consider solutions in $(\mathbb Z_{\geq0})^3\backslash\{(0,0,0)\}$ to
\begin{equation}\label{markoffeq}
x_1^2+x_2^2+x_3^2-3x_1x_2x_3=0.
\end{equation}
Equation (\ref{markoffeq}) is known as the \emph{Markoff equation}, and its solutions are called \emph{Markoff triples}, with the integers that occur as members of some triples known as Markoff numbers. As we discuss below, a lot is known about Markoff numbers, and one particularly useful observation, both to us and more generally in the arithmetic study of Markoff numbers, is that one can generate all such triples by considering the orbit of the group generated by the involutions
$$R_1(x_1,x_2,x_3)=(3x_2x_3-x_1,x_2,x_3)$$
$$R_2(x_1,x_2,x_3)=(x_1,3x_1x_3-x_2,x_3)$$
$$R_3(x_1,x_2,x_3)=(x_1,x_2,3x_1x_2-x_3)$$
acting on the triple $(1,1,1)$ \cite{M1},\cite{M2}. In this way, one can view the set of Markoff triples as a tree as depicted in Figure~\ref{Markofftree}. Note that the tree depicted in Figure~\ref{Markofftree} shows one of several very similar branches of the tree, the others are generated by acting on $(1,1,1)$ via $R_1$ and $R_2$, as well as letting one other involution act on the triple immediately adjacent to $(1,1,1)$ (in the figure, that other involution would be $R_1$ acting on $(1,1,2)$). Those branches will simply contain permutations of the triples shown in Figure~\ref{Markofftree}.

\begin{figure}[h!] 
    \centering
    \includegraphics[scale=.45]{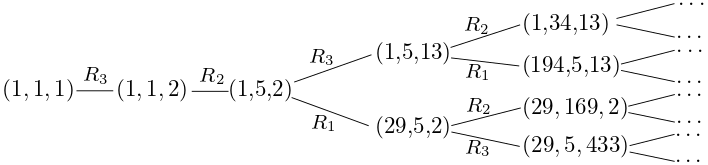}
    \caption{A branch of the Markoff tree generated by applying the involutions $R_1,R_2,R_3$ to the fundamental solution (1,1,1).}
    \label{Markofftree}
\end{figure}

Markoff triples first appeared in the literature in Markoff's master's thesis \cite{M1}, \cite{M2} in the context of studying rational approximations via continued fractions. Markoff found that the sequence of Markoff numbers plays a big role in results that produce infinitely many irrationals $\alpha$ which admit a continued fraction convergent $p/q$ such that
$$\left|\alpha-\frac{p}{q}\right|<\frac{m}{\sqrt{d}q^2}$$
for appropriate values of $m$ and $d$ which are connected to the Markoff sequence. Moreover, he was able to determine exactly for which $\alpha$ the above bound is sharp. Markoff's work on the subject inspired a whole series of generalizations of his result which introduced novel techniques into the theory which are still used today, e.g. \cite{RS}. Indeed, today the Markoff equation in (\ref{markoffeq}) is of interest not only to those studying continued fractions but has become an important object in other branches of mathematics such as algebraic geometry and theoretical physics \cite{Ch},\cite{LT}.

For cryptographic applications, we will consider what is, roughly speaking, the mod-$p$ reduction of this tree, as well as a related graph where edges are defined slightly differently.

Specifically, let $p$ be a (large) prime, and consider the set of nonzero solutions modulo $p$ to equation (\ref{markoffeq}). We call a solution $(x_1,x_2,x_3)$ in $(\mathbb F_p)^3$ a \emph{triple}, and each entry in the triple $x_1$, $x_2$, or $x_3$, a \emph{coordinate}; so a coordinate is simply an element of $\mathbb F_p$.

We consider two graphs: $G_p$ and $\hat{G}_p$. In both of these graphs, the vertices are comprised of nontrivial (we exclude $(0,0,0)$) solutions modulo $p$. In $G_p$, the edges are defined by the involutions $R_1,R_2,R_3$: two triples are connected by an edge if one of the three involutions takes one triple to the other. We will also refer to this graph as the involution graph. In $\hat{G}_p$, the edges are defined by  \emph{rotations} (see Section \ref{rotsection}). Explicitly, they are given by
\begin{equation}\label{rotations}
\rot_i=\tau_{i+1,i+2}\circ R_{i+1}
\end{equation}
Here $\tau$ is a transposition of coordinates, and all index additions are done modulo 3. Two triples are connected by an edge in $\hat{G}_p$ if one of the 3 rotations takes one triple to the other. We refer to this as the rotation graph.

Our reason for considering $G_p$ is that it is particularly convenient for setting up our hash function. It is the rotation graph $\hat{G}_p$, however, for which Bourgain-Gamburd-Sarnak prove connectivity and in fact give a path-finding algorithm. Notably, finding paths in the graph $\hat{G}_p$ is easily correlated to finding paths in $G_p$, and vice versa.

Specifically, one can check that, given three different indices $1\leq i,j,k\leq 3$, and a triple $(a,b,c)$, one has $R_i(a,b,c)=\tau_{j,i}R_j\tau_{k,j}R_k\tau_{i,k}R_i(a,b,c)$, so a path of length $\ell$ between two triples in $G_p$ corresponds to a path of length $3\ell$ in $\hat{G}_p$, whereas a path of length $\ell$ between two triples in $\hat{G}_p$ corresponds to a path of at length at least $\ell/3$, and possibly much longer, in $G_p$. In other words, an algorithm to find paths in $\hat{G}_p$ will find paths in the other graph in time not significantly shorter, and possibly much longer, time.


\begin{figure} 
    \centering
    \includegraphics[scale=.13]{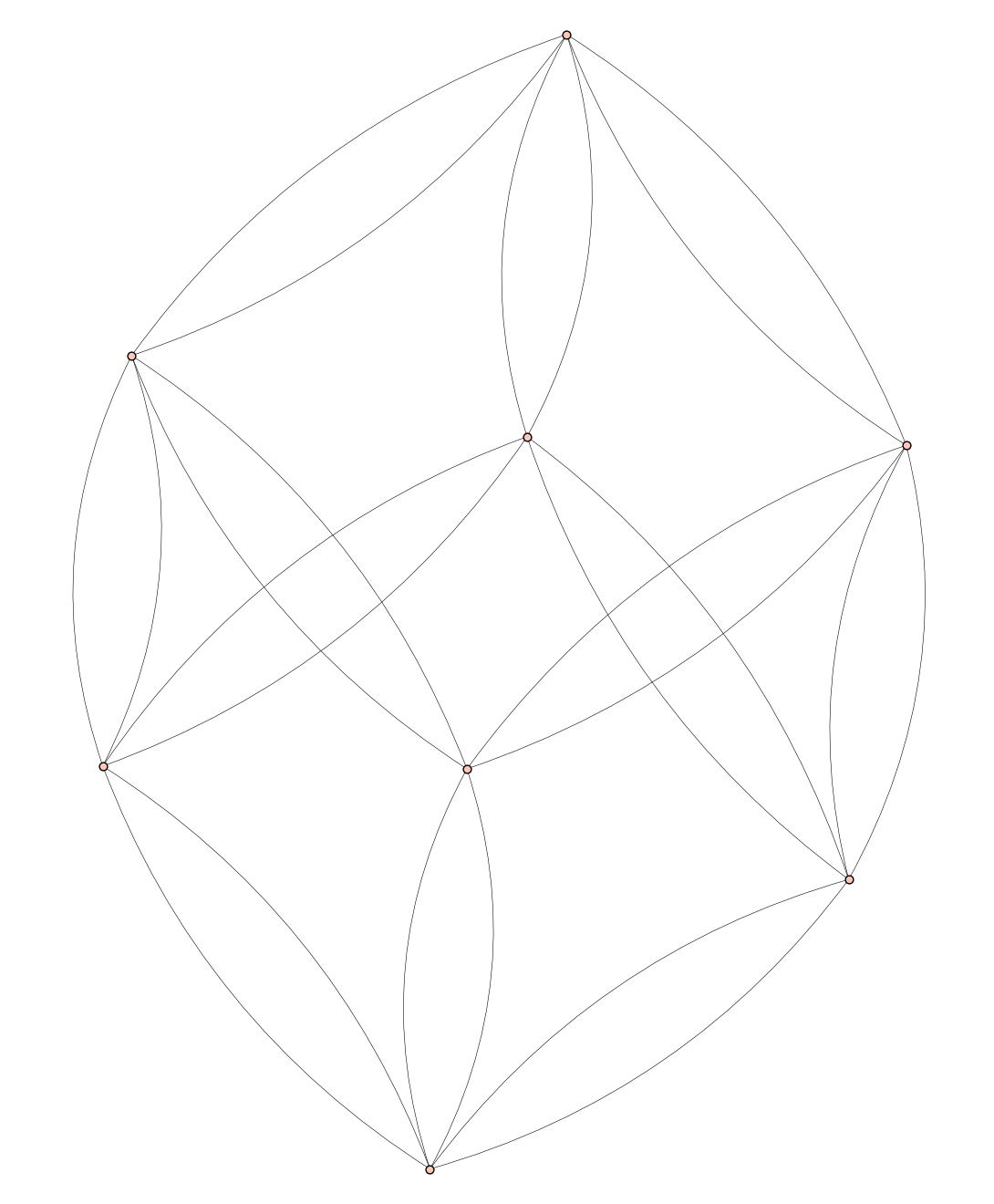}\quad
    \includegraphics[scale=.13]{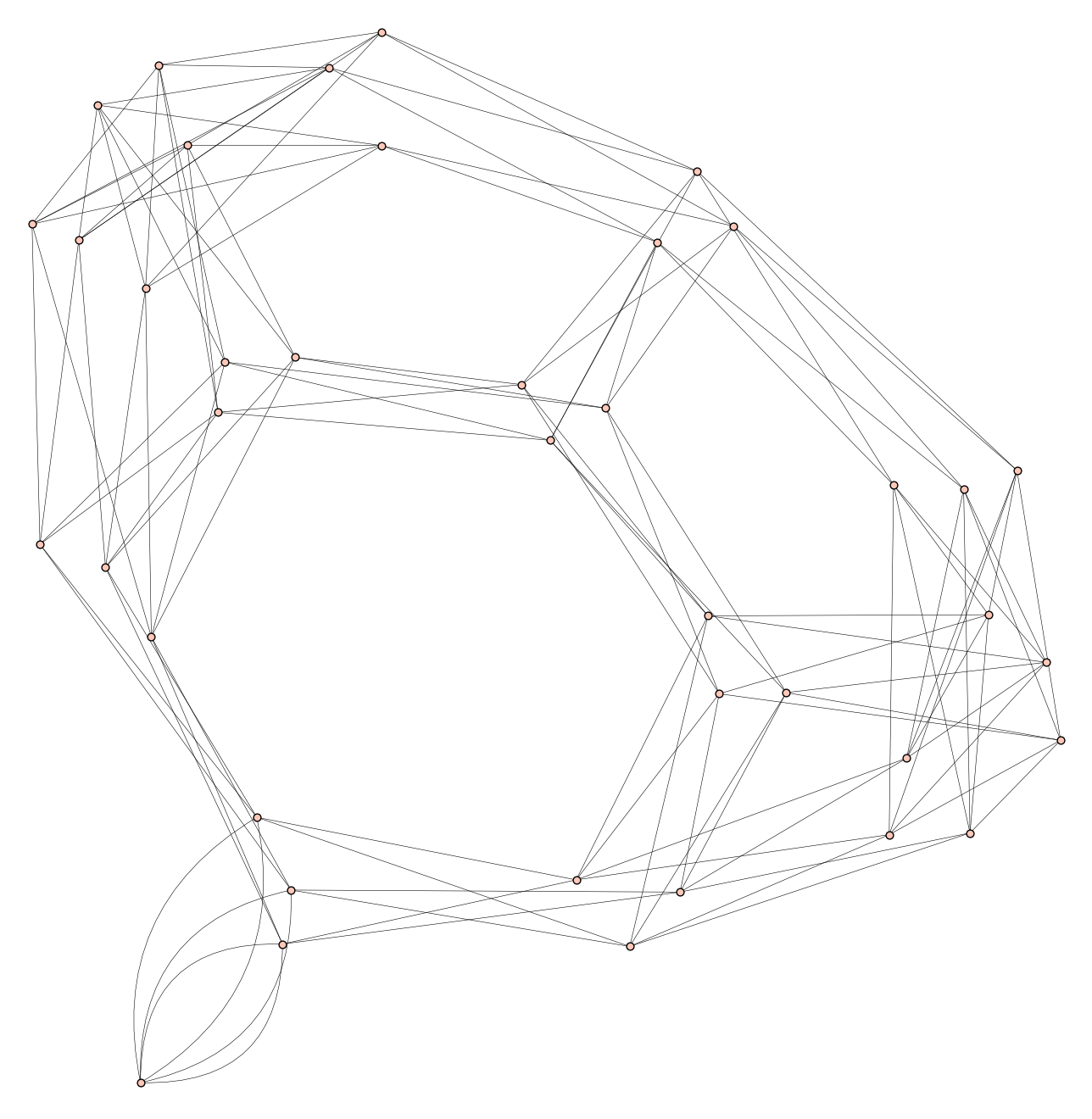}\quad
    \includegraphics[scale=.13]{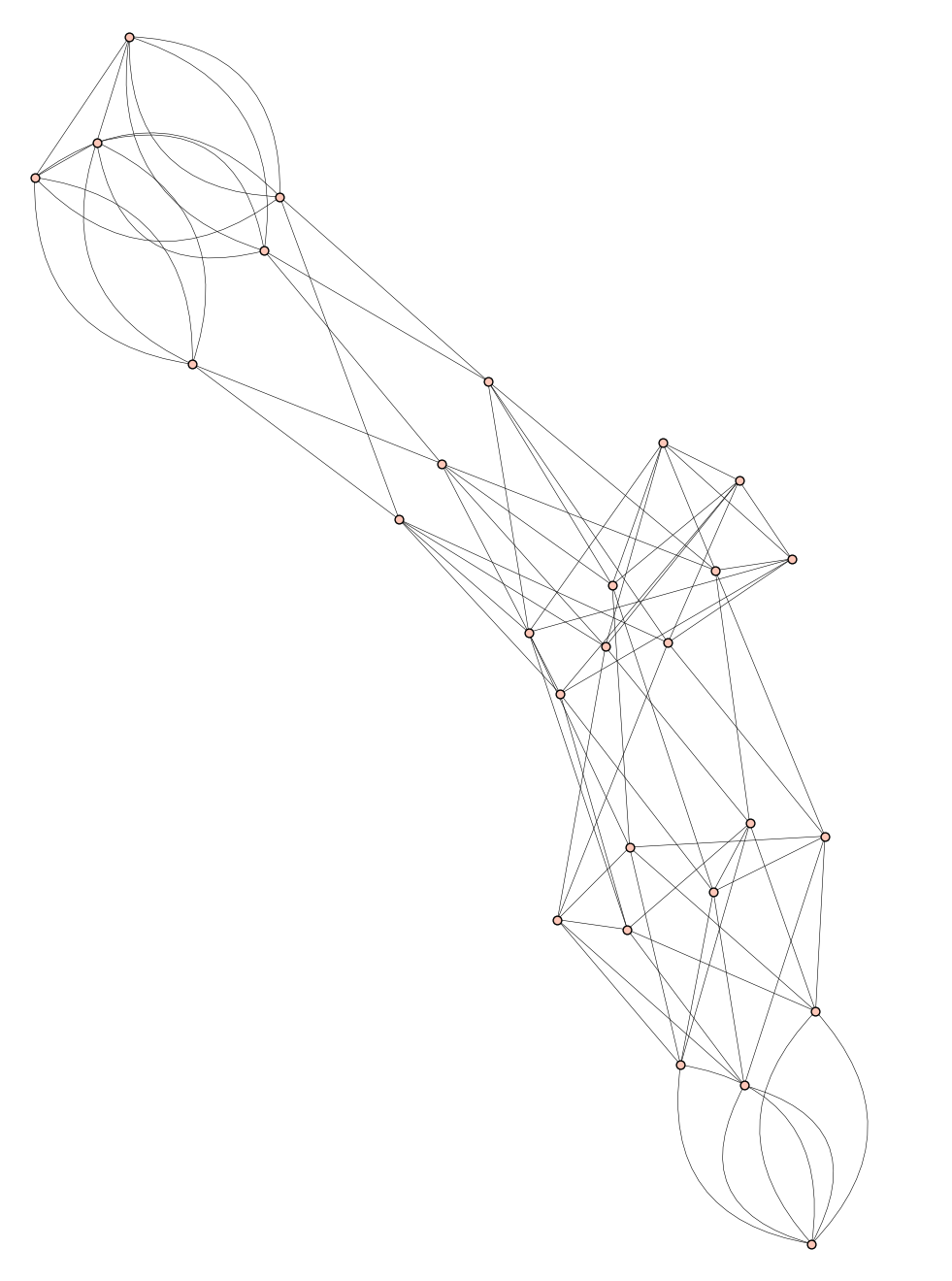}
    \caption{The Markoff mod-$p$ graphs $\hat{G}_p$ for $p=3$, $5$, and $7$.}
    \label{Markofftreemodp}
\end{figure}


\subsection{Cryptographic hash function}
First we give a brief summary of the hash function we propose, with more details given in Section~\ref{cryptsection}.
We choose a large prime $p$: typical cryptographic size primes have at least $256$ bits. There are some restrictions on the choice of $p$, see the discussion in the following section. We then label in $G_p$ the edges corresponding to the involutions $R_1$, $R_2$, and $R_3$, respectively. The input to the hash function is a bit string $b_0b_1b_2\ldots$.  To compute the output of the hash function, start at a fixed vertex, such as the vertex $(1,1,1)$, and take a walk in the graph according to the directions in the bit string, reading the string bit by bit, one bit for each step of the walk. The output is the vertex at which this process ends.

The security of this hash function depends on the hardness of finding paths between two vertices in the graph $G_p$. This paper is concerned primarily with exploring potential avenues for attack and estimating their running times.

\subsection{Avenues for attack}
In Sections~\ref{BGSattack} and \ref{Attack2}, we describe two potential attacks on our hash function. Both give ways of finding paths between two triples in the graph $G_p$.

The work of Bourgain, Gamburd, and Sarnak in \cite{BGS1} gives a potential attack on the cryptosystem we have proposed, but we will show that its running time is heuristically $O(p)$. Specifically, in order to show that the graphs $\hat{G}_p$ are connected, their algorithm gives a way to construct a path between any two vertices in the Markoff graph mod $p$ for most primes $p$. The very rough idea is to associate to every vertex a certain ``order" that we describe in Section \ref{rotsection}. Those vertices of maximal possible order are grouped into what Bourgain-Gamburd-Sarnak call \emph{the cage}, which they prove is connected. They then show a way to connect any other vertex to the cage by walking along the graph while increasing the so-called order of the vertex until the cage is reached. We call this the BGS algorithm for finding paths in $\hat{G}_p$.

Not surprisingly, this is easiest to do when the order of the vertex in question is already quite large, and they call the process of passing from such a vertex to the cage the ``End Game." Those points whose order is larger than a small power of $p$ but not as large as the vertices involved in the End Game are treated separately in what Bourgain, Gamburd, and Sarnak call the ``Middle Game."  Finally, to move from vertices of very small order into the Middle Game and beyond, they employ rather technical methods in what they call the ``Opening." Because it is central to understanding the potential attack that it gives on our cryptosystem, and because some of the details that are important to us are left to the reader in \cite{BGS1}, we carefully describe these three parts of Bourgain, Gamburd, and Sarnak's proof of Theorem \ref{BGSmain} and then give a heuristic for the running time in Section \ref{cryptheur}. We show the following.

\begin{prop}\label{pathlengthprop}
The length of the path given by the BGS algorithm is bounded from above by
$$O\of{p\log\log{p}}.$$
\end{prop}
See section~\ref{cryptheur} for the derivation. From this, we deduce the time complexity of the algorithm.

\begin{prop}\label{timecomplexityBGS}
The time complexity of the BGS algorithm is at most
$$O\of{p\log\log{p}}.$$
\end{prop}

We also provide data in Section \ref{cryptheur} showing that these upper bounds are in fact of the right order, which supports the heuristics. We observe that the path is longest modulo primes $p$ for which $p^2-1$ is \emph{very smooth} (see discussion in Section \ref{smoothdef}). Note that, if a faster way of finding paths between vertices in $G_p$ or $\hat{G}_p$ exists, this would be both of interest to the study of our hash function, and to the study of the arithmetic of Markoff numbers: it could provide a second proof of Theorem \ref{BGSmain} that works even for those primes that Bourgain, Gamburd, and Sarnak could not handle.

Related to this first attack, we also present a seemingly simple attack which would use the fact that, as one walks without back-tracking along the Markoff tree depicted in Figure~\ref{Markofftree}, the coordinates of the triple increase. This makes it trivial to find paths between vertices in this tree: simply walk from the two vertices in question along edges that decrease coordinates until one gets to $(1,1,1)$. If it were easily possible, given a vertex $(x_1,x_2,x_3)$ in $G_p$ to lift it to a vertex in the infinite Markoff tree which reduces to $(x_1,x_2,x_3)$ modulo $p$, then one could simply connect the two vertices in question by connecting their lifts in the infinite tree, and then transferring this path back to $G_p$. However, in Section \ref{Attack2}, where we describe this attack more carefully, we explain the obstacles to this approach. In particular, an efficient algorithm to lift would yield another proof of Theorem \ref{BGSmain}: an algorithm to lift any triple in $G_p$ to one in the infinite tree (which is known to be connected) in fact shows that $G_p$ itself, and, by the previous discussion of how the two graphs are related, $\hat{G}_p$ is connected. We conjecture the following. 

\begin{conj}\label{LiftingConjecture} The length of a path found by lifting a triple in $G_p$ to a triple in the Markoff tree over $\mathbb Z$ is at least $\textrm{O}(\log p)$ for most triples in $G_p$.
\end{conj}

Note that this does not take into account the difficulty of actually finding the lift  (or this path of length $\log p$). Hence the running time of this attack is likely comparable to the one based on the BGS algorithm.

\subsection{Some background on Markoff triples}
We now note a few important facts about the graphs $G_p$. First of all, it is known that $|G_p|=|\hat{G}_p|=p^2+\left(\frac{-1}{p}\right)\cdot 3p$ if $p>3$, which is mentioned in \cite{dCM} without proof. Here $\left(\frac{\ast}{\ast}\right)$ denotes the Legendre symbol. One way to prove this is to think of the left side of the Markoff equation as a quadratic form in one of the variables $x_1,x_2,x_3$, and then consider how many representations of $0$ there are mod $p$, which is a well known problem. Furthermore, Meiri and Puder have proven the following.
\begin{thm}[Meiri, Puder \cite{MP}]
Let $G_p$ be as above and let $\Gamma_p$ be the finite permutation group induced by the action of $\Gamma=\langle R_1,R_2,R_3\rangle$ on $G_p$. Then, outside a zero-density subset of all primes, the group $\Gamma_p$ is either the full symmetric group or the alternating group on the vertices of $G_p$.  
\end{thm}
They conjecture that this is in fact true for all primes $p\geq 5$. So we can compare walking along the graph $G_p$ to generating elements of $S_n$ and $A_n$ (where $n=|G_p|$) with a given random generating set, which has been studied, for example in \cite{BH}.

The graphs $G_p$ and $\hat{G}_p$ are now known to be connected for the majority of primes $p$. This was proven by Bourgain, Gamburd, and Sarnak in \cite{BGS1} as a first step in studying the arithmetic of Markoff triples (for example, the distribution of primes or numbers with a bounded number of prime factors among Markoff numbers). The structure of these graphs plays an important roll in sieving over Markoff triples, which is key in \cite{BGS2}. Specifically, they show the following.
\begin{thm}[Bourgain, Gamburd, Sarnak \cite{BGS1}]\label{BGSmain}
For all primes $p\not\in E$, where $E$ is an exceptional set of primes, the graph $\hat{G}_p$ is connected. The set $E$ is small: for any $\epsilon>0$, the number of primes $p\leq T$ with $p\in E$ is at most $T^{\epsilon}$ for $T$ large.
\end{thm}

Furthermore, they conjecture not only that $G_p$ is connected for all primes $p$, but that in fact the family of graphs $G_p$ where $p$ is prime is an expander family.  This is explored in \cite{dCM}, as we discuss in the following section. Theorem~\ref{BGSmain} is enough to show, as Bourgain, Gamburd, and Sarnak show in \cite{BGS2}, that the set of Markoff numbers contains infinitely many composite numbers, and in fact that almost all Markoff numbers are composite.


\section{Markoff triple hash function and data}\label{cryptsection}
Recall that for a sufficiently large prime $p$, we can construct a hash function as follows.
A fixed public initial vertex is specified, say $(1,1,1)$. Also choose an involution $k$; the choice of $k$ is fixed but arbitrary.
The edges of $G_p$ are canonically labeled with 1, 2, or 3, corresponding to the three involutions $R_1,R_2,R_3$ respectively. Given a bit string of finite length as input, say $b_0b_1b_2\ldots$, designate $c_0=k$.  Then, for $i>0$, suppose $c_{i-1}$ was the label of the previous edge, $c_i\in\{1,2,3\}$. Then we move along the edge
$$c_i=(c_{i-1}+b_i)\Mod{3}+1$$
Note that doing this avoids substrings of the form $R_iR_i$, and so we avoid backtracking. The output of the hash function is the final vertex where the walk ends, after processing all the bits $b_i$ in the string. Note that the initial bit string is not necessarily raw text or data, and will most likely be augmented with some compression function, such as the Merkle-Damgard construction.

For example, suppose we want to encode the binary message \verb|10011| in $G_{13}$. We choose $k=0$ then apply the series of rotations
$$10011\mapsto R_2\circ R_3\circ R_2\circ R_1\circ R_2(1,1,1)=(0,5,1)$$
We know that $|G_p|=O(p^2)$, so for the output space of this hash function to be comparable to say SHA-256, we would want to take $p\approx2^{128}$.
The security of this hash function depends upon the difficulty of path or cycle finding in $G_p$. That is, given $x$, $y\in G_p$, what is the time complexity of finding a path between $x$ and $y$?

Note that if the starting vertex is $(1,1,1)$, the input string needs to be longer than $\log{p}$ so that the coordinates of the output start to wrap around modulo $p$.  Otherwise a trivial lifting attack is possible.  A better starting vertex $v_0$ can be obtained by taking a walk of length $\log(p)$ from $(1,1,1)$.  In general we will assume that the length of the walk from $v_0$ is at least length $\log(p)$. In fact, \cite{BGS2} has conjectured that the family $G_p$ is an expander family; walks of length $O(\log{p})$ are sufficient for mixing in expander graphs.


\subsection{Cryptographic Heuristics}\label{cryptheur}
The theorems of Bourgain, Gamburd, and Sarnak, which we present in detail in Section~\ref{BGSattack} prove the correctness of the following path finding algorithm in $\hat{G}_p$ (under certain easy assumptions on $p$). This path finding algorithm uses a notion of ``order" of a triple, coming from a certain rotation assigned to it (see Section \ref{rotsection} for the definition). The idea is then that there is a large connected component of $\hat{G}_p$ consisting of triples of ``maximal" order, and to connect any two triples one need only connect each of them to this large component, which Bourgain-Gamburd-Sarnak call the \emph{cage}. One does this by walking along a specially concocted path in which the orders of the triples grow as one walks along it, until one reaches the cage. In other words, the algorithm runs as follows.

Suppose we want to connect two triples $X$ and $Y$. We can do this in two steps:
\begin{enumerate}[leftmargin=0cm]
\item First, if $X$ or $Y$ are not in the cage, then we want to connect them to the cage.

Every triple $X$ is part of special cycles in $\hat{G}_p$ which we describe in Section~\ref{BGSattack}, called maximal orbits $M_X$ of $X$. Bourgain-Gamburd-Sarnak show that the orbit $M_X$ contains at least one point of higher order than $X$, call it $X'$. Then $X'$ is connected to $X$, so replace $X$ with $X'$ and repeat the same argument. The order is guaranteed to increase each step, until eventually the order is maximal.
\item Now we can suppose $X$ and $Y$ are both in the cage. Then by Proposition \ref{incprop} (Proposition~6 in \cite{BGS1}), there exists a point $Z$ in the cage such that $X-Z-Y$ is a valid path.

In fact, we have an explicit way of finding $Z$. Since $X$ or $Y$ might have more than one maximal orbit, we search over all maximal orbits of $X$ and $Y$ and look for an intersection, which is guaranteed to exist. In the case that $X$ and $Y$ have the same singular maximal index, then we simply perform an appropriate transposition on either $X$ or $Y$.
\end{enumerate}

As noted in Proposition~\ref{pathlengthprop}, we have an upper bound of
$$O(p\log\log{p})$$
on the length of the path obtained using the BGS algorithm. 

To see this, suppose we start at a point that is far from the cage and walk towards it by going around essentially a full orbit of a rotation acting on that point, then switching to an orbit that is slightly larger, and so on. The number of steps needed is as follows: we may have to go through all the divisors $d$ of $p-1$ and $p+1$ as we increase the divisors (these correspond to the orbit sizes). We only need to do this up to about $\sqrt p$, according to Bourgain-Gamburd-Sarnak, since after that there will be a point in the orbit that is in the cage. So the total number of steps needed to take in the BGS algorithm is bounded above by $$\sum_{d|p-1} d + \sum_{d|p+1} d +3p$$ where $3p$ bounds the number of steps needed to walk from one point in the cage to another. We have that $$\sum_{d|p-1} d,\;  \sum_{d|p+1} d <<p\log\log p$$ and this is expected to be almost always sharp.

Note that if the two points between which wee must find a path are both in the cage, then path-finding in the cage has complexity $O(p)$.

We can imagine optimizing this algorithm by being greedier with the first step. Instead of looking at the entire orbit, as soon as we find \emph{any} $X'$ with order higher than $X$, we replace $X$ with that $X'$. The algorithm is also guaranteed to work because the order is still guaranteed to increase at each step. If we assume this $X'$ occurs uniformly randomly within the orbit, instead of looking at $d$ points in an orbit, we only look at $d/2$ points on average. The complexity of this modified algorithm is largely unchanged (up to a constant).

This upper bound is supported by the data in Figure \ref{logpvslogavgtime}. Larger $p$ are needed for more precise comparison, but this is encouraging data that tracks very closely with $p\log\log p$. Note that the relative scale for time taken is arbitrary; nevertheless we mention the specifications for reference. Calculations were done in SageMath 9.1, running on a quad-core i7-8550U CPU.
\begin{figure}[h!]
\centering
\includegraphics[scale=.5]{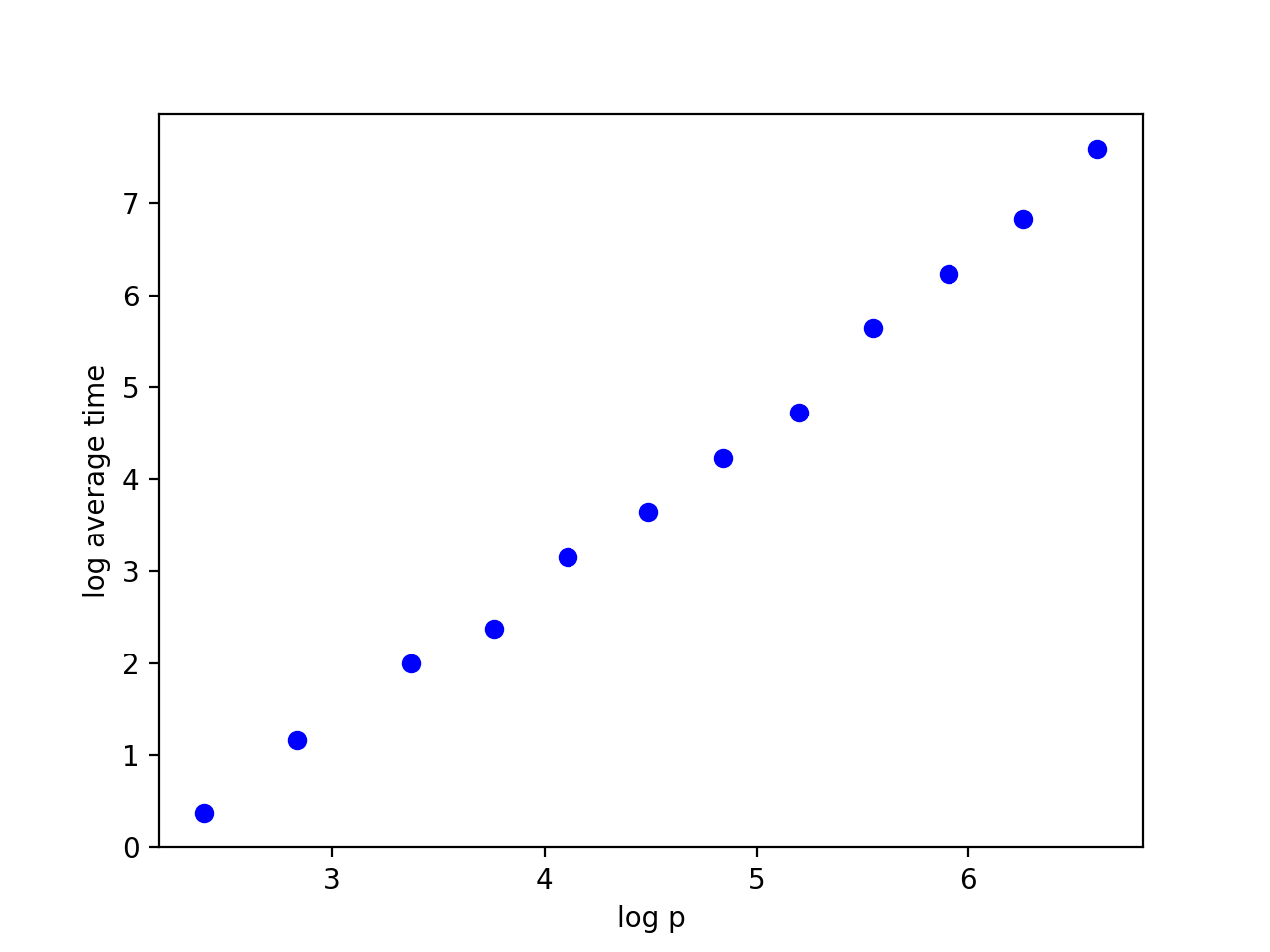}
\caption{Plot of $\log{p}$ vs. $\log$ average time taken by the BGS algorithm. Here we take primes $p\le739$. Time taken is in milliseconds, averaged over 100 trials for each $p$.}
\label{logpvslogavgtime}
\end{figure}

Now the complexity depends merely on the chance of a worst-case scenario, where the triple is not in the cage. This depends on a couple of factors: the proportional size of the cage, and the number of steps it could potentially take to connect any triple to the cage.

It turns out that both of these factors depend in turn on $\eta_p$, the  number of divisors of $p^2-1$. There is a correlation between $\eta_p$ and the number of steps needed to connect a triple (not in the cage), as can be seen in Figure \ref{logvsavgtime}. Additionally, Figure \ref{logvsprop} is supporting evidence that the size of cage also depends on $\eta_p$. The asymptotic behavior of this graph, as $\eta_p\to\infty$, is a relevant open question.
\begin{figure}[h!]
\centering
\includegraphics[scale=.5]{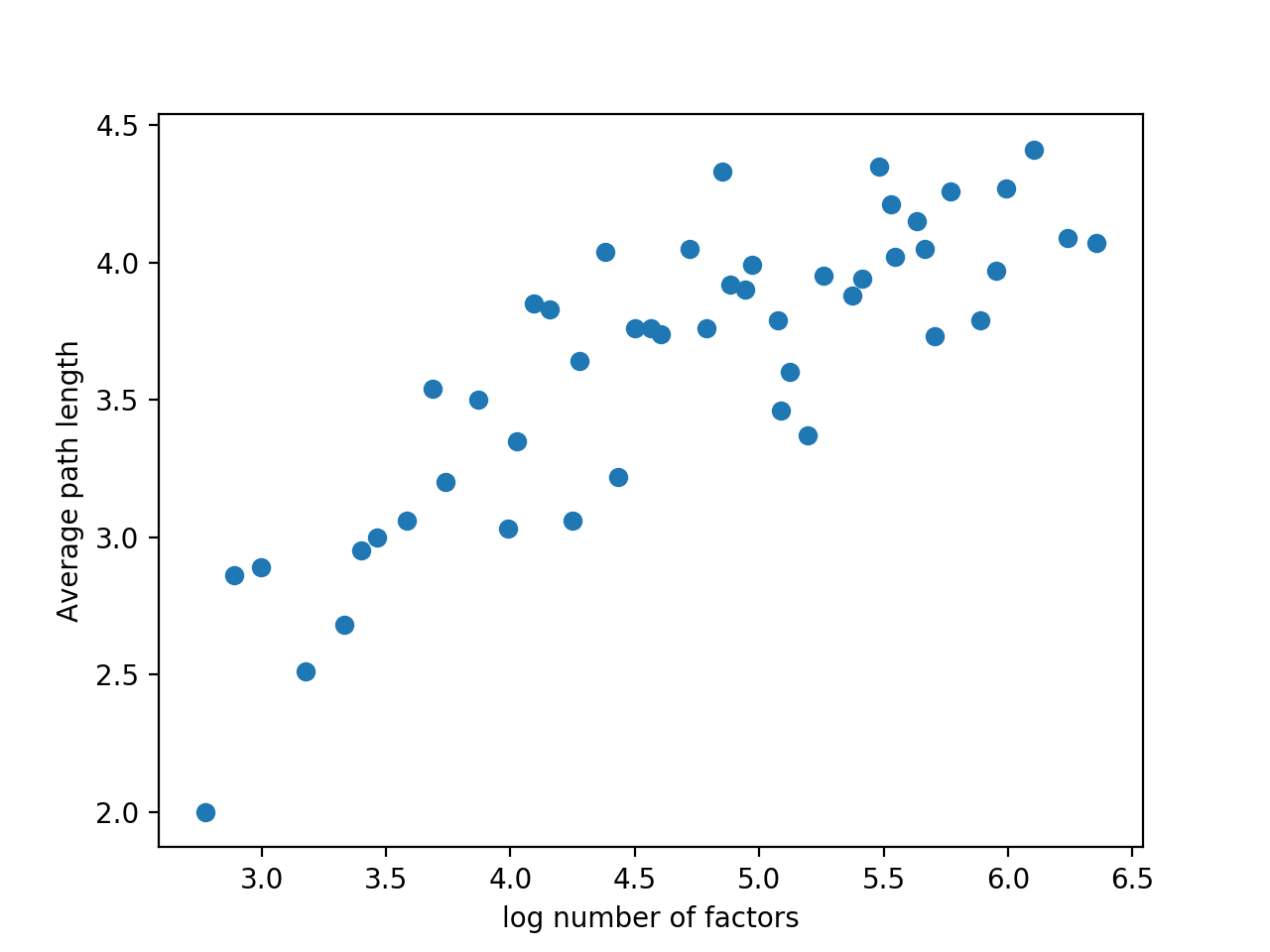}
\caption{Plot of $\log\eta_p$ vs. average time taken by the BGS algorithm, in seconds. Here we take prime $p<10000$ and $\eta_p<6000$. Time taken is in seconds, averaged over 10 trials for each $p$.}
\label{logvsavgtime}
\end{figure}

\begin{figure}[h!]
\centering
\includegraphics[scale=.5]{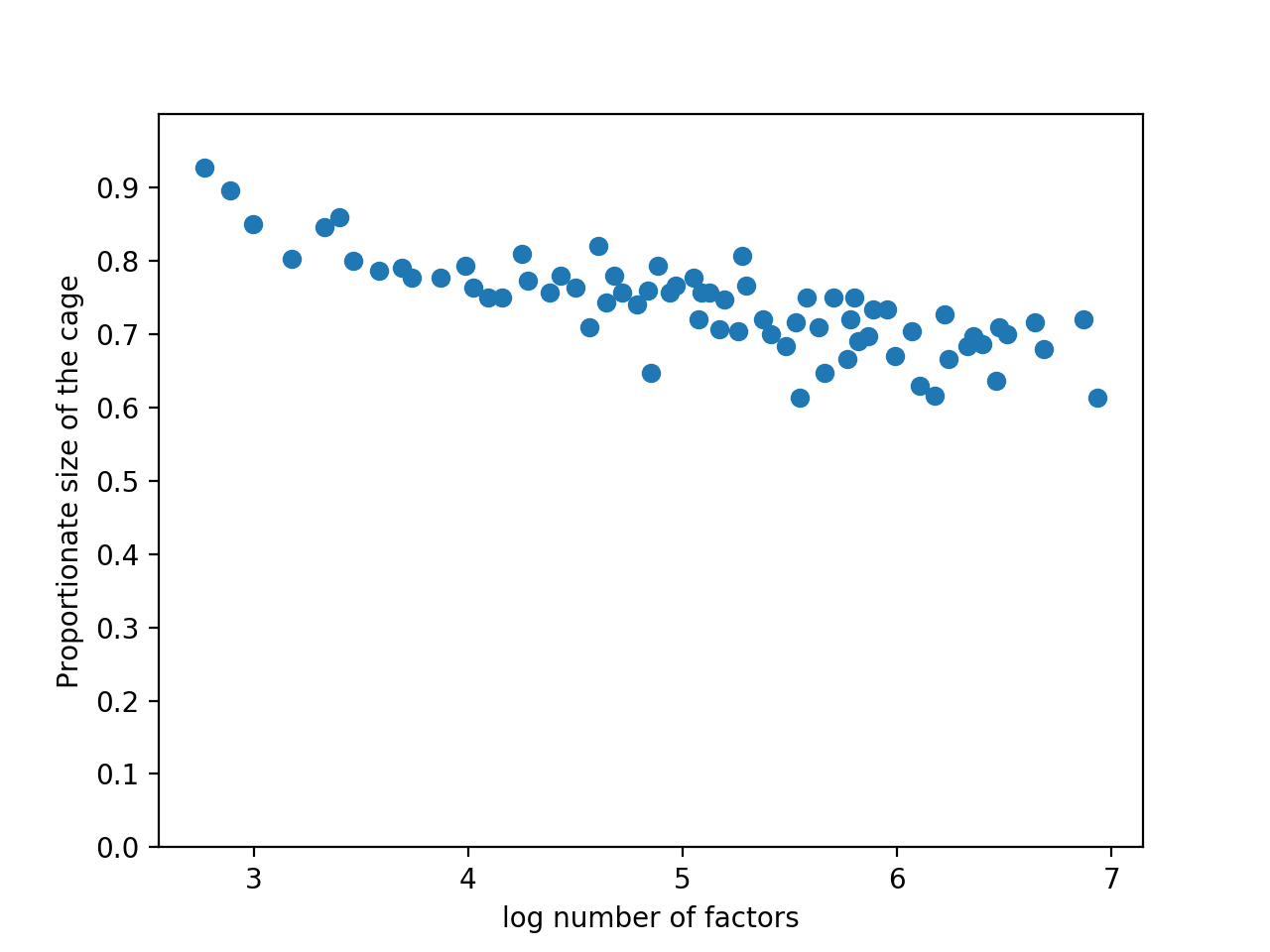}
\caption{Plot of $\log\eta_p$ vs. proportion of all vertices in the graph $\hat{G}_p$ which are in the cage for primes $p<100000$ with $\eta_p<24000$.}
\label{logvsprop}
\end{figure}

We would also like to see if the time taken to connect a point to the cage depends on $\eta_p$. A plot of this relationship can be seen in Figure \ref{fig8}, showing a strong correlation between $\eta_p$ and the time taken to connect a point to the cage.

\begin{figure}[h!]
\centering
\includegraphics[scale=.5]{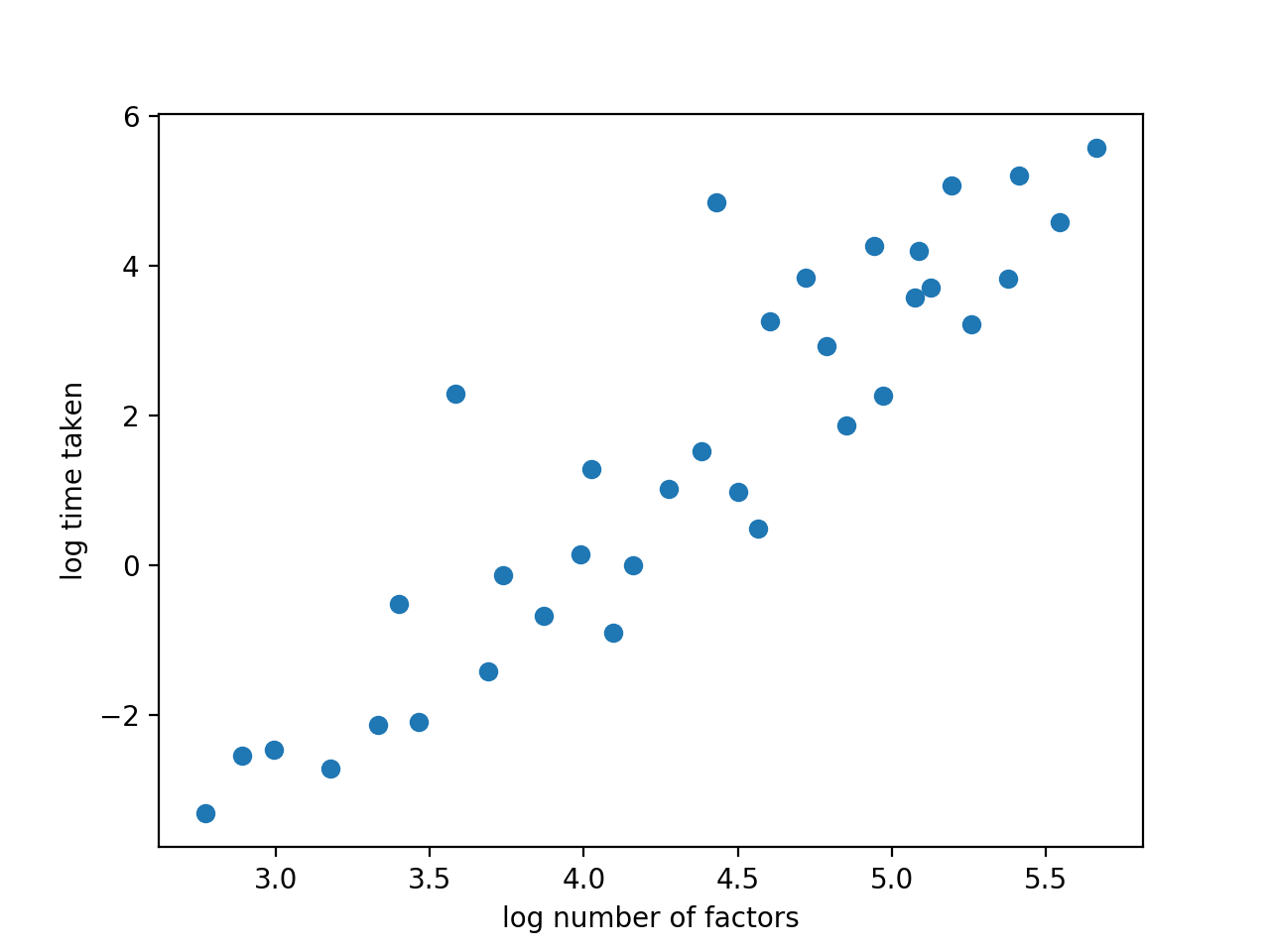}
\caption{Plot of $\log\eta_p$ vs. time taken by the BGS algorithm to connect a point to the cage, in seconds. Here our primes $p$ are taken such that $p<2000$ and $\eta_p<2000$. Time taken is in seconds, averaged over 10 trials for each $p$.}
\label{fig8}
\end{figure}

Concretely, \cite{BGS1} only establishes that $\hat{G}_p$ is connected as long as $p$ satisfies the following condition: for any $y$, \label{smoothdef}
$$\sum_{d\mid p^2-1,d\in[(\log p)^{1/3},y]}d^{2/3}<y$$
Therefore $p$ is selected so that $p^2-1$ is not smooth. Not only does this guarantee connectedness, it also assists with the problem of short cycles. The length of any orbit must divide $p^2-1$, so avoiding small factors will also avoid small orbits. Fortunately, such primes are difficult to find, and thus easy to avoid. Different search methods have been proposed for finding smooth primes (for example, in the appendices of \cite{C} and \cite{FKLPW}, the authors produce two separate approaches), all of which support the claim that finding such primes is a difficult task. 

This is additional evidence that increasing $\eta_p$ also increases the difficulty of path-finding. Thus we recommend that the security parameter be dependent on both the size of $p$ as well as $\eta_p$.

\subsection{Sampling}
Without prior knowledge of the entirety of $\hat{G}_p$, how does one randomly sample a point from $\hat{G}_p$? One way would be the following. Start at a fixed point, say $(1,1,1)$ which is in $\hat{G}_p$ for all $p$. Then perform a random non-backtracking walk starting from $(1,1,1)$, of a large length $l$, the end of which is our sample.
Below we see empirically that $l$ does not affect the random distribution for sufficiently large $l$:

\begin{figure}[h!]
\centering
\includegraphics[scale=.45]{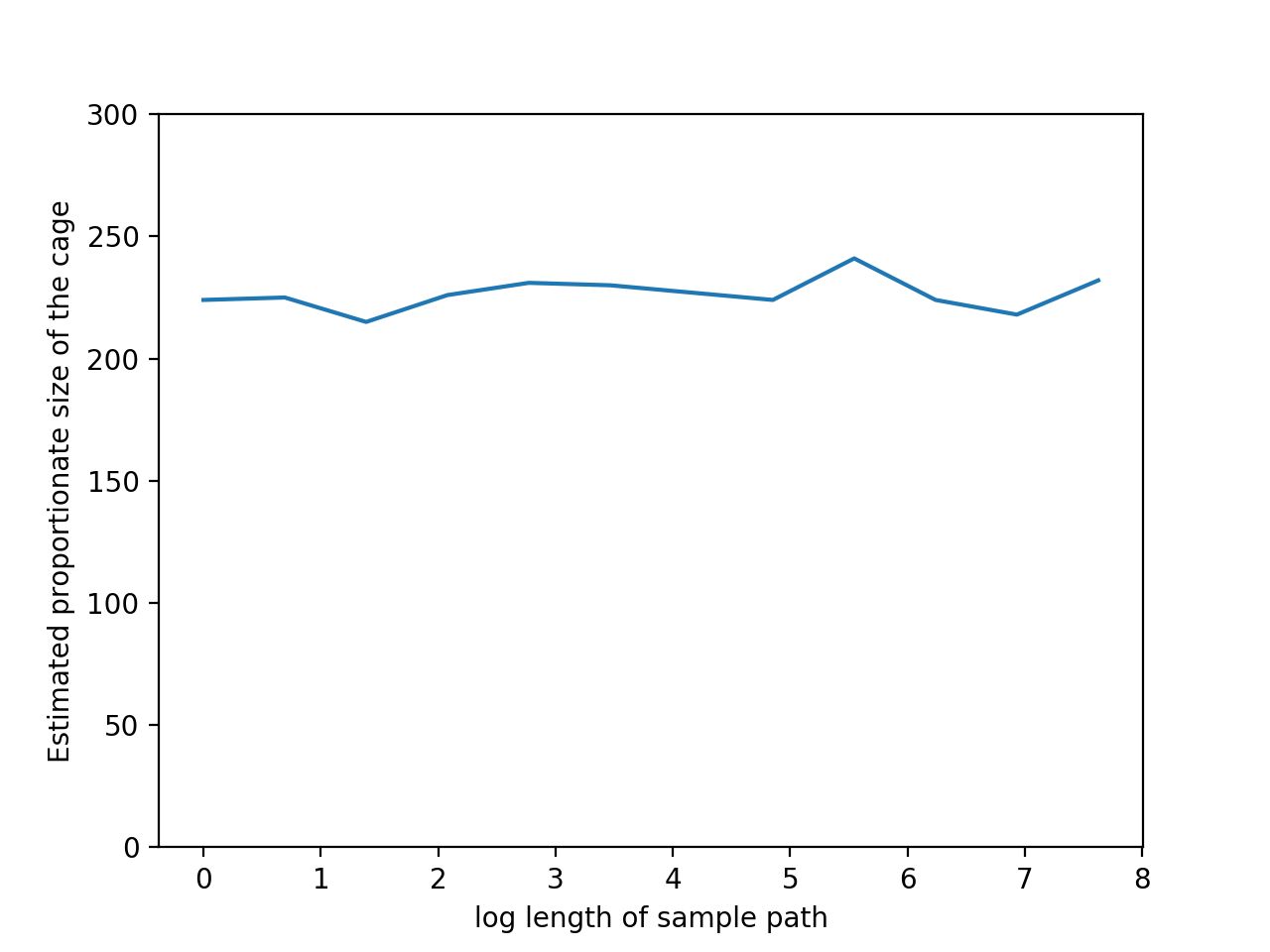}
\caption{Graph of $\log$ length of a sample path vs. a sampled estimate of the size of the cage. Here $p$ is fixed at 5851, and we performed 300 random walks of length $l$ for each $l$. The data show how many of these 300 samples are in the cage.}
\end{figure}

Of course, this method would be truly uniformly random if the family of graphs $\hat{G}_p$ were an expander family. We do have empirical evidence of this, as well as more compelling evidence from a paper of de Courcy-Ireland and Magee \cite{dCM}. In particular, they state that $G_p$ ``resembles" a random graph, which is a start to examining the spectral gap for the adjacency matrices of the graphs $G_p$.

Specifically, de Courcy-Ireland and Magee show that the distribution of the eigenvalues of the adjacency matrix of a Markoff graph $G_p$ asymptotically follows the Kesten-McKay law for the distribution of eigenvalues of large randomly chosen $3$-regular graphs. In general, the Kesten-McKay law says that, for a large random $d$-regular graph, the expected eigenvalue probability distribution is
$$\rho_d(\lambda)=\frac{d}{2\pi}\frac{\sqrt{4(d-1)-\lambda^2}}{d^2-\lambda^2}$$
for $|\lambda|\leq 2\sqrt{d-1}$ and is $0$ otherwise. 

Let $\mu_p$ be the distribution of eigenvalues on $G_p$, which range from $[-3,3]$:
$$\mu_p=\frac{1}{|G_p|}\sum\delta_{\lambda_j}$$
De Courcy-Ireland and Magee prove the following.
\begin{thm}[\cite{dCM}, Theorem 1.1]
Given $p$, there exists a constant $L\sim\log{p}$ and a constant $C$, independent of $p$ and $L$, such that
$$\int x^Ld\mu_p=\int x^L\rho_3(x)dx+O(C^L/p)$$
\end{thm}
However, as mentioned in \cite{dCM}, this distribution is not strong enough to show that the family of graphs $G_p$ is an expander family. We would like the spectral gap to be nonzero, i.e. the number of eigenvalues in the interval $[3-\epsilon,3]$ to be $O(1)$; the work in \cite{dCM} only proves that this number is $O(p^2/\log{p})$.

For a beautiful graphical comparison of the plot of this distribution to analogous plots of calculated eigenvalues for the Markoff surface mod $p$ for $p=83$ and $89$, see Figure 1.1 of \cite{dCM}. While not conclusive, this gives some indication that the family of Markoff graphs forms an expander family.





\section{Attack by Pathfinding via the Method of Bourgain-Gamburd-Sarnak}\label{BGSattack}

In this section, we go through the key elements of the proof of Theorem~\ref{BGSmain}, which is necessary for the analysis of how fast of a path-finding algorithm this produces in section~\ref{cryptheur}.

\subsection{Rotations}\label{rotsection}
A key collection of tools in the proof of Theorem~\ref{BGSmain} are certain rotations that are associated to every triple in $\hat{G}_p$. In this section, we go over crucial results about these rotations.

Denote by $\tau_{ij}$ the transposition of the $i$th and $j$th coordinates.

Let $C_j(a)$ denote all triples for which the $j$th coordinate is equal to $a$.

Given a triple $X$, define a \emph{rotation} function
$$\rot_{x_1}(X)=\tau_{23}\circ R_2(X)=(x_1,x_3,3x_1x_3-x_2)$$
Note that this function is easily extended to $x_2$ or $x_3$ by applying the appropriate permutation to $R_j(X)$ for the appropriate $j$.

Further note that without loss of generality we applied $R_2$ instead of $R_3$; we can simulate the latter by again applying the appropriate permutation to $X$. Thus define $\rot_{x_2},\rot_{x_3}$ similarly.

Since $\rot_{x_1}$ fixes $x_1$, we can think of $\rot_{x_1}$ as a function in $(x_2,x_3)$ on the plane defined by setting the first coordinate to be $x_1$:
$$\rot_{x_1}\begin{pmatrix}x_2\\x_3\end{pmatrix}=\begin{pmatrix}0&1\\-1&3x_1\end{pmatrix}\begin{pmatrix}x_2\\x_3\end{pmatrix}$$
So we define the \emph{rotation order} of $x$ as the order of 
$$\begin{pmatrix}0&1\\-1&\ x\end{pmatrix}\in \SL_2(\F_p)$$
Note here we are replacing $3x_1$ with $x$. For the remainder of this section, $x$ will always denote $3x_1$.

The rotation order of a triple is then defined to be the maximum rotation order of its coordinates.

Iteratively applying one rotation to a triple $X$ eventually returns one to $X$, and we call the set of all such points the \emph{orbit} of that rotation.

The eigenvalues of the rotation matrix are $\frac{x\pm\sqrt{x^2-4}}{2}$, and so we separate cases depending on whether $\leg{x^2-4}{p}=\pm1$.
\begin{itemize}
\item If $x\equiv\pm2\Mod{p}$, we say $x$ is \emph{parabolic}.
\item If $\leg{x^2-4}{p}=1$, then $x$ is \emph{hyperbolic}.
\item If $\leg{x^2-4}{p}=-1$, then $x$ is \emph{elliptic}.
\end{itemize}
A triple is parabolic/hyperbolic/elliptic if its coordinate with maximal rotation order is parabolic/hyperbolic/elliptic. These suggestive names will begin to make more sense if $\hat{G}_p$ is pictured literally as a subset of Euclidean space.

To reiterate an above statement, if we were to say $12\in\F_{17}$ is hyperbolic, we mean that the coordinate $x_1=4$ is hyperbolic.
\begin{lem}[Lemma 3 of \cite{BGS1}]\label{parabolicels}
Let $x$ be parabolic, i.e. $x\equiv\pm2\Mod{p}$. If $p\equiv3\Mod{4}$, then $C_1(x)$ is empty (i.e. $x$ does not appear in any triple in $\hat{G}_p$). If $p\equiv1\Mod{4}$, then
$$C_1(2/3)=\of{\frac{2}{3},t,t\pm\frac{2i}{3}}$$
$$C_1(-2/3)=\of{-\frac{2}{3},t,-t\pm\frac{2i}{3}}$$
where $i^2\equiv-1\Mod{p}$ and $t$ is any number $\Mod{p}$. So $C_1(x)$ is a pair of disjoint lines. Furthermore, the action of $\rot_x$ is explicitly given by
$$\rot_x\of{\of{\frac{2}{3},t,t\pm\frac{2i}{3}}}=\of{\frac{2}{3},t\pm\frac{2i}{3},t\pm\frac{4i}{3}}$$
$$\rot_x\of{\of{-\frac{2}{3},t,-t\pm\frac{2i}{3}}}=\of{-\frac{2}{3},-t\pm\frac{2i}{3},-t\mp\frac{4i}{3}}$$
So $\rot_2$ fixes each line while $\rot_{-2}$ interchanges them.
\end{lem}
\begin{proof}
Without loss of generality suppose $x_1=\pm2/3$. Then equation (\ref{markoffeq}) reduces to
$$x_2^2+x_3^2+\frac{4}{9}\mp2x_2x_3\equiv0\Mod{p}$$
$$(x_2\mp x_3)^2\equiv-\frac{4}{9}\Mod{p}$$
So a solution exists if and only if $\leg{-1}{p}=1$, which is equivalent to $p\equiv1\Mod{4}$. 

Set $p\equiv1\Mod{4}$ and suppose $C_1(2/3)=(2/3,t,t+a)$. Then we have
$$t^2+(t+a)^2+\frac{4}{9}-2t(t+a)\equiv0\Mod{p}$$
which reduces to $a^2\equiv-4/9\Mod{p}$ independent of $t$, which gives the desired result.

Similarly suppose $C_1(-2/3)=(-2/3,t,-t+a)$, which gives
$$t^2+(-t+a)^2+\frac{4}{9}+2t(-t+a)\equiv0\Mod{p}$$
which again reduces to $a^2\equiv-4/9\Mod{p}$ which again gives the desired result.

Now we can explicitly calculate
$$\rot_2\of{\of{\frac{2}{3},t,t\pm\frac{2i}{3}}}=\of{\frac{2}{3},t\pm\frac{2i}{3},3\frac{2}{3}\of{t\pm\frac{2i}{3}}-t}=\of{\frac{2}{3},t\pm\frac{2i}{3},t\pm\frac{4i}{3}}$$
\begin{align*}\rot_{-2}\of{\of{-\frac{2}{3},t,-t\pm\frac{2i}{3}}}&=\of{-\frac{2}{3},-t\pm\frac{2i}{3},-3\frac{2}{3}\of{-t\pm\frac{2i}{3}}-t}\\&=\of{-\frac{2}{3},-t\pm\frac{2i}{3},-t\mp\frac{4i}{3}}\end{align*}
as desired.
\end{proof}
Here is an example in $G_{17}$, where $17\equiv1\Mod{4}$, $2/3\Mod{p}=12$, and $i\equiv4\Mod{17}$.
\begin{figure}
    \centering
    \includegraphics[scale=.17]{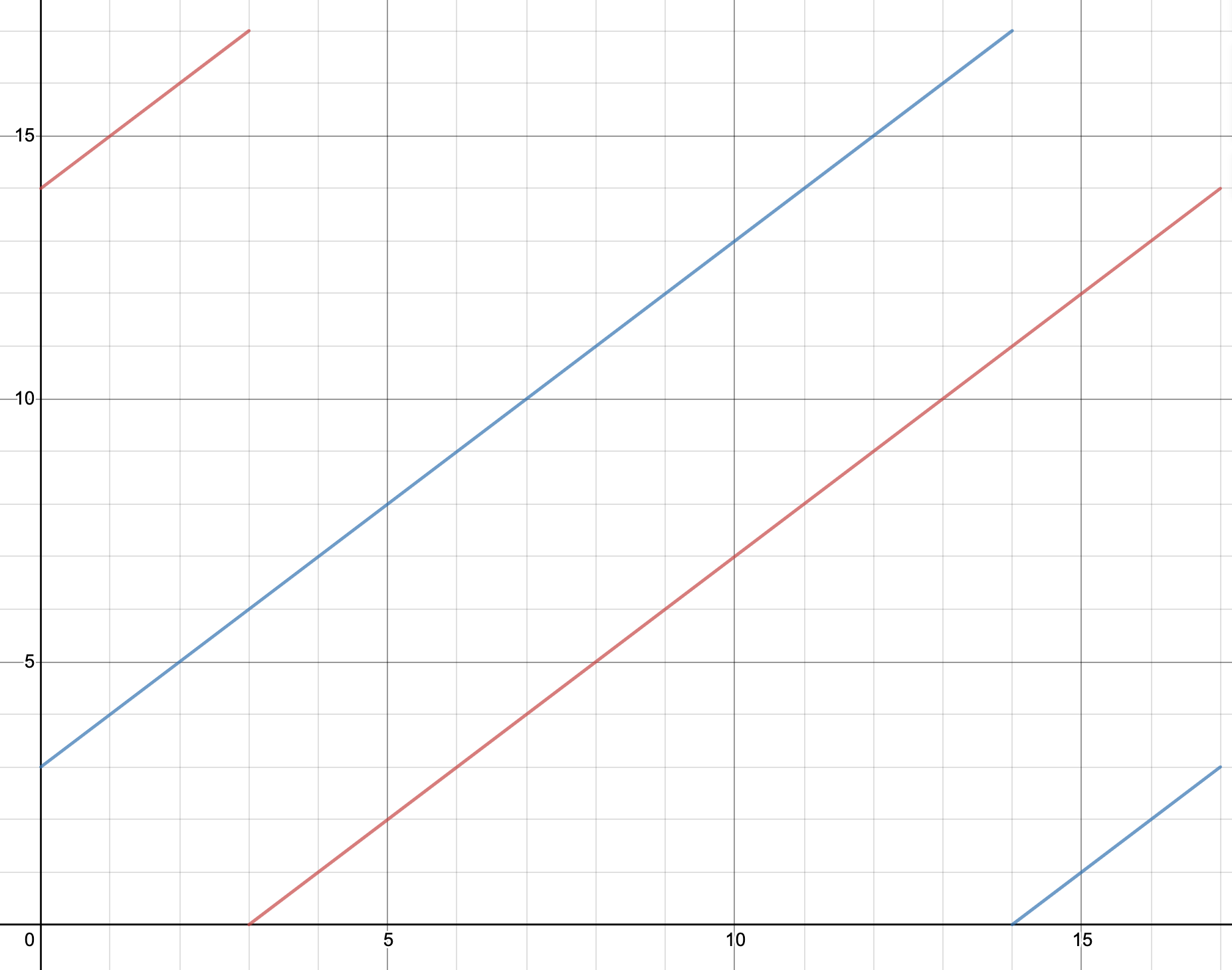} \qquad
    \includegraphics[scale=.17]{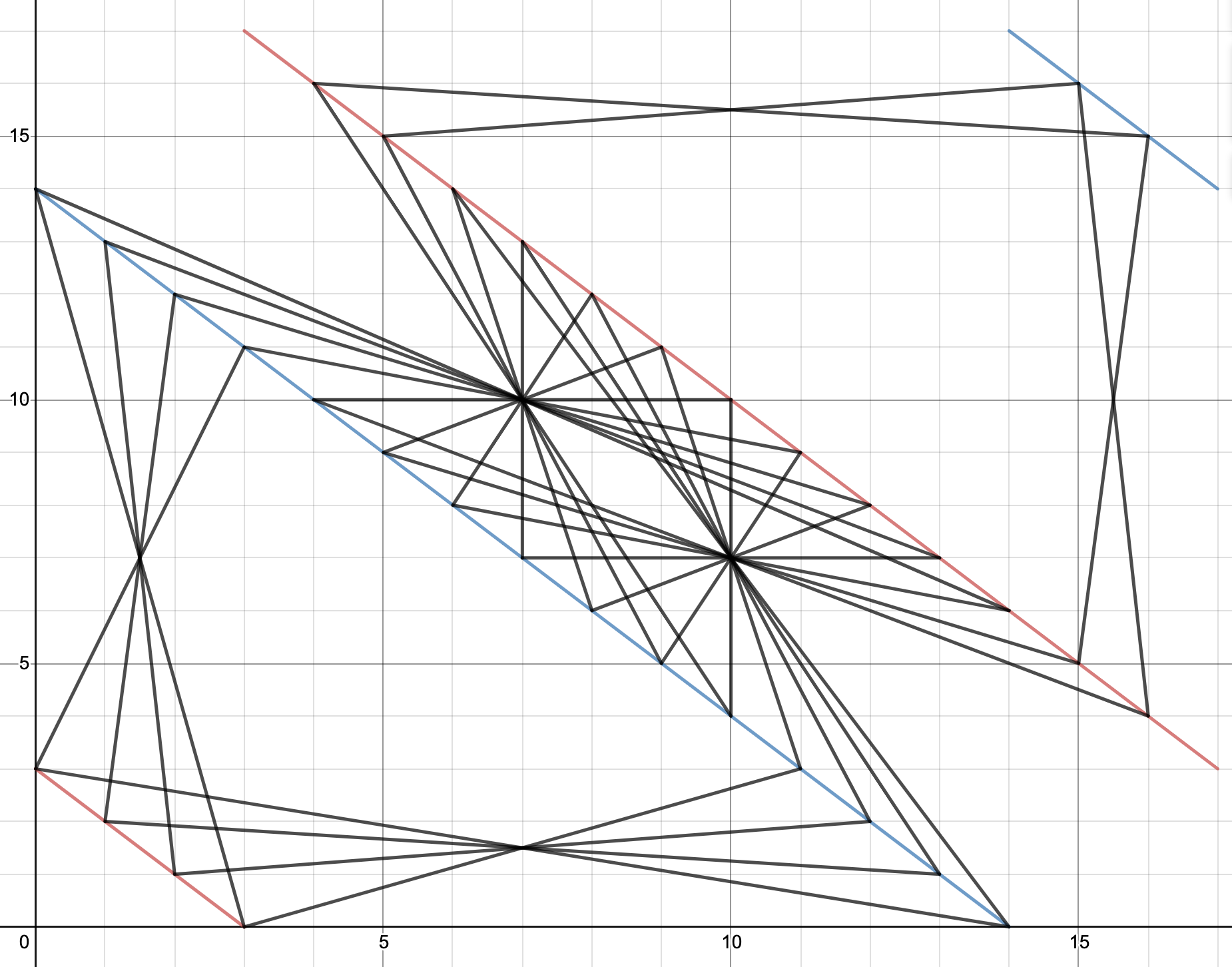}
    \caption{(a) Two lines in $C_1(12)$ fixed by $\rot_{12}$. (b) Two lines in $C_1(5)$ interchanged by $\rot_5$. }
    \label{fig:my_label}
\end{figure}


\begin{lem}\label{plusinverse}
If $x$ is not parabolic, then we can write
$$x=\chi+\chi^{-1}$$
where $\chi\in\F_p$ if $x$ is hyperbolic, and $\chi\in\F_{p^2}$ if $x$ is elliptic.
\end{lem}
\begin{proof}
Suppose $x$ is hyperbolic, i.e. $\leg{x^2-4}{p}=1$. Then suppose $x^2-4=r^2$. Set $\chi=(x+r)/2$, $\chi^{-1}=(x-r)/2$, and verify
$$\frac{x+r}{2}\frac{x-r}{2}=\frac{x^2-r^2}{4}=1$$
Then we have
$$r^2=(\chi-\chi^{-1})^2=\chi^2+\chi^{-2}-2=(\chi+\chi^{-1})^2-4$$
as desired. Now similarly suppose $x$ is elliptic. Then $x^2-4$ is not a residue in $\F_p$, but it is a residue in 
$$\F_{p^2}\simeq\frac{\F_p[y]}{y^2-(x^2-4)}$$
Set $x^2-4=r^2$ where $r\in\F_{p^2}$ and repeat the above argument.
\end{proof}

Upon diagonalizing the rotation matrix $\rot_x$, one arrives at
\begin{align*}
\rot_x&=\begin{pmatrix}1&1\\ \chi&\chi^{-1}\end{pmatrix}
\begin{pmatrix}\chi&0\\ 0&\chi^{-1}\end{pmatrix}
\begin{pmatrix}1&1\\ \chi&\chi^{-1}\end{pmatrix}^{-1}\\
&=(\chi^{-1}-\chi)^{-1}
\begin{pmatrix}1&1\\ \chi&\chi^{-1}\end{pmatrix}
\begin{pmatrix}\chi&0\\ 0&\chi^{-1}\end{pmatrix}
\begin{pmatrix}\chi^{-1}&-1\\ -\chi&1\end{pmatrix}
\end{align*}
Thus
\begin{align*}
(\rot_x)^l&=(\chi^{-1}-\chi)^{-1}
\begin{pmatrix}1&1\\ \chi&\chi^{-1}\end{pmatrix}
\begin{pmatrix}\chi^l&0\\ 0&\chi^{-l}\end{pmatrix}
\begin{pmatrix}\chi^{-1}&-1\\ -\chi&1\end{pmatrix}\\&=(\chi^{-1}-\chi)^{-1}
\begin{pmatrix}\chi^{l-1}-\chi^{1-l}&-\chi^l+\chi^{-l}\\ \chi^l-\chi^{-l}&-\chi^{l+1}+\chi^{-l-1}\end{pmatrix}
\end{align*}
If we consider $\chi^l=t$, where $t\in\ip{\chi}$, then we have
$$\ip{\rot_x}=\left\{(\chi^{-1}-\chi)^{-1}\begin{pmatrix}\chi^{-1}t-\chi t^{-1}&t^{-1}-t\\t-t^{-1}&\chi^{-1}t^{-1}-\chi t\end{pmatrix}:t\in\ip{\chi}\right\}$$
Thus
\begin{align}
C_1(x)&=\left\{(\chi^{-1}-\chi)^{-1}\begin{pmatrix}\chi^{-1}t-\chi t^{-1}&t^{-1}-t\\t-t^{-1}&\chi^{-1}t^{-1}-\chi t\end{pmatrix}\begin{pmatrix}x_2\\x_3\end{pmatrix}:t\in\ip{\chi}\right\}\nonumber\\
&=(\chi-\chi^{-1})^{-1}\Big(t(x_3-\chi^{-1}x_2)+t^{-1}(\chi x_2-x_3),\\
&\quad t(\chi x_3-x_2)+t^{-1}(x_2-\chi^{-1}x_3)\Big)\label{matrixfactor}
\end{align}
again for $t\in\ip{\chi}$. Now we can rewrite the second coordinate as $at+bt^{-1}$ where
$$a=\frac{x_3-\chi^{-1}x_2}{\chi-\chi^{-1}},\quad b=\frac{\chi x_2-x_3}{\chi-\chi^{-1}}$$
Later we will need the fact that
\begin{equation}\label{notone}ab=\frac{x_2x_3(\chi+\chi^{-1})-x_2^2-x_3^2}{(\chi-\chi^{-1})^2}=\frac{x^2}{(\chi-\chi^{-1})^2}=\of{\frac{\chi+\chi^{-1}}{\chi-\chi^{-1}}}^2\ne1\end{equation}
Now we consider the cases of $x$ hyperbolic or elliptic separately.
\begin{itemize}[leftmargin=0cm]
\item For $x$ hyperbolic: From equation (\ref{matrixfactor}), note that $a,b\in\F_p^*$, so substitute $t\mapsto ta^{-1}$ to see that
$$C_1(x)=\left\{\of{t+\frac{ab}{t},\chi t+\frac{ab}{\chi t}}:t\in\F_p^*\right\}$$
Applying the rotation gives
\begin{equation}\label{hyperbolicmap}\rot_x\of{t+\frac{ab}{t},\chi t+\frac{ab}{\chi t}}=\of{\chi t+\frac{ab}{\chi t},\chi^2 t+\frac{ab}{\chi^2 t}}\end{equation}
Since $t\in\F_p^*$, we see that $|C_1(x)|=p-1$. On the other hand, since $x$ is hyperbolic, by Lemma \ref{plusinverse}, we can write
$$x=\rho^j+\rho^{-j}$$
where $\rho$ is a primitive root of $\F_p$. Then if we iteratively apply $\rot_x$, we cycle through $\frac{p-1}{j}$ elements in $C_1(x)$, i.e. the rotation order of $\rot_x$ is $\frac{p-1}{j}$ for some $j$.

An explicit example of this can be seen below in Figure \ref{fig2} for the case of $G_{17}$:
\begin{figure}
    \centering
    \includegraphics[scale=.16]{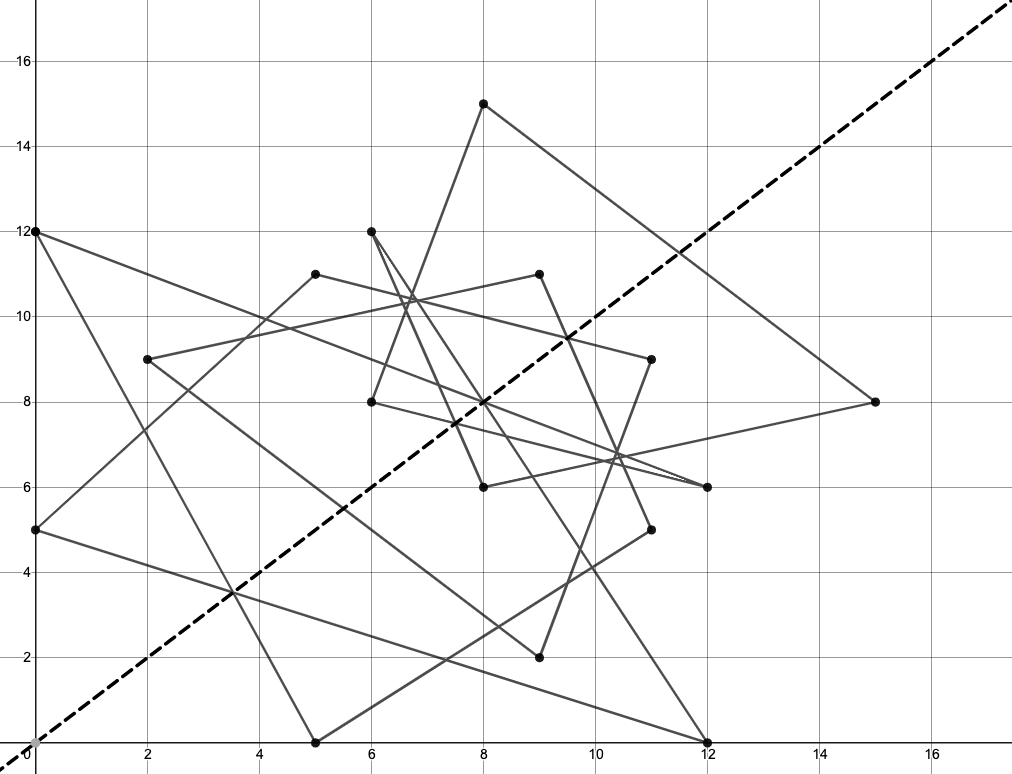} \qquad
    \includegraphics[scale=.16]{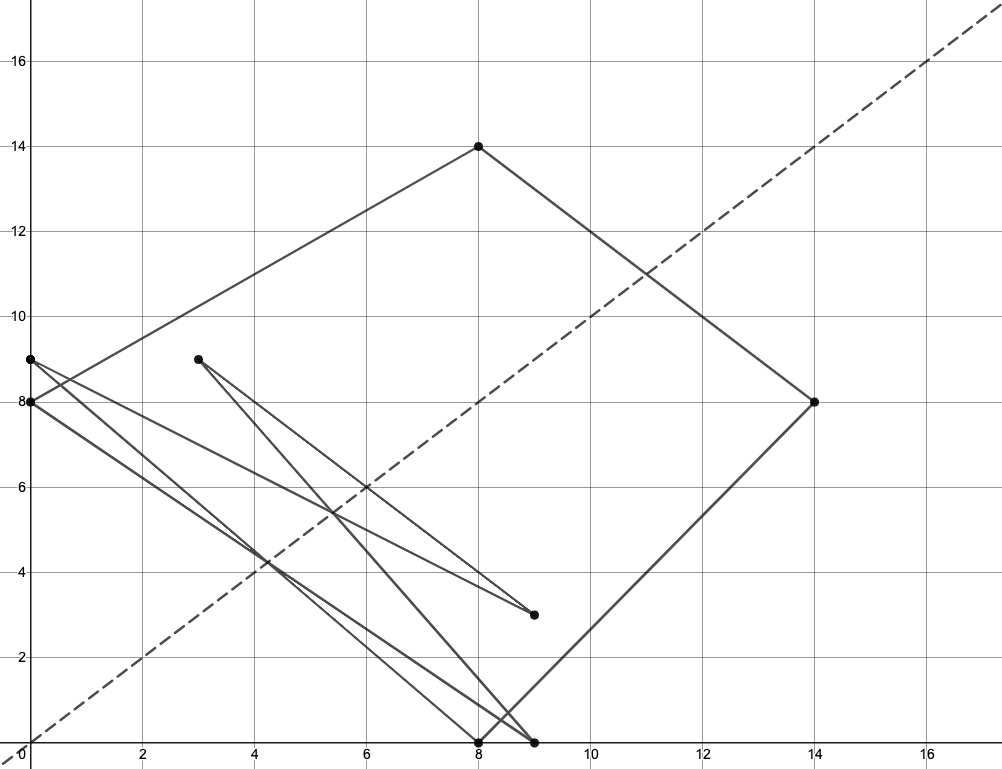}
    \caption{(a) A maximal order hyperbolic rotation $\rot_3$ shown in the plane $C_1(3)$. (b) A hyperbolic rotation $\rot_2$ of order $8=\frac{p-1}{2}$ in $C_1(2)$. The line $x_2=x_3$ is shown for symmetry.}
    \label{fig2}
\end{figure}

\item For $x$ elliptic, the derivation is similar. We start by rewriting
$$x=\chi+\chi^{-1}=\nu+\nu^p$$
for $\nu\in\F_{p^2}-\F_p$. Then applying the rotation gives
\begin{equation}\label{ellipticmap}\rot_x\of{x,t,\frac{\kappa_x}{t}}=\of{x,t+\frac{\kappa_x}{t},t\nu+\frac{\kappa_x}{t\nu}}\end{equation}
which implies
$$C_1(x)=\left\{\of{t+\frac{ab}{t},\nu t+\frac{ab}{\nu t}}:t\in\F_{p^2}^*,\quad t^{p+1}=ab\right\}$$
where the latter requirement implies $t\in\F_{p^2}\backslash\F_p$, which in turn implies $|C_1(x)|=p+1$. On the other hand, since $x$ is elliptic, by Lemma \ref{plusinverse}, we can write
$$x=\xi^j+\xi^{-j}$$
where $\xi$ is some element of $\F_{p^2}$. Explicitly, if $\gamma$ is a generator of $(\F_{p^2})^\times$, then $\xi=\gamma^{p+1}$. So if we iteratively apply $\rot_x$, we cycle through $\frac{p+1}{j}$ elements in $C_1(x)$, i.e. the rotation order of $x_1$ is $\frac{p+1}{j}$ for some $j$.

An explicit example of rotations for elliptic and hyperbolic elements in $G_{17}$ can be seen in Figure \ref{fig3}.
\begin{figure}[h!]
    \centering
    \includegraphics[scale=.16]{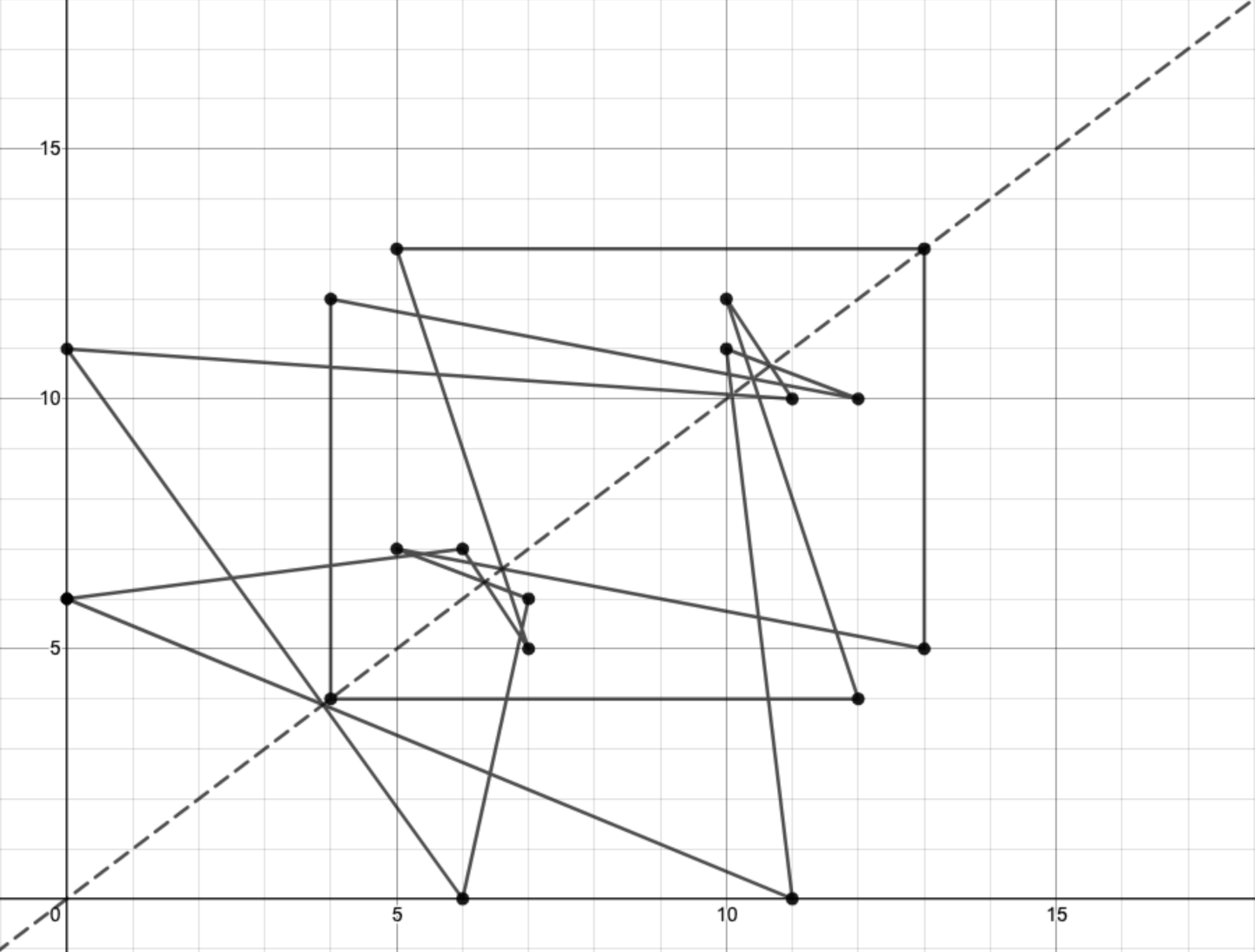} \qquad
    \includegraphics[scale=.16]{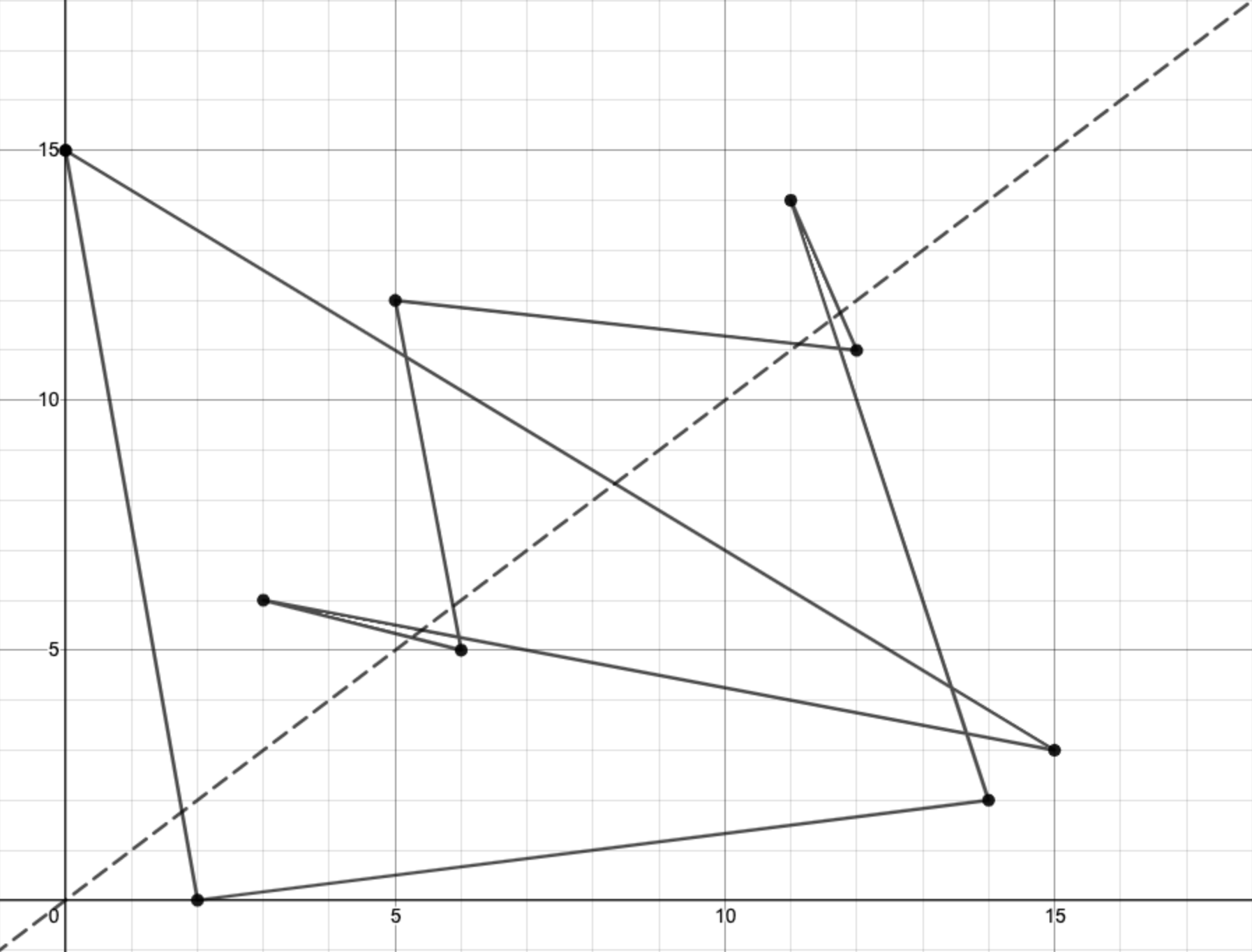}
    \caption{(a) A maximal order elliptic rotation $\rot_7$ shown in the plane $C_1(7)$. (b) A hyperbolic rotation $\rot_8$ of order $9=\frac{p+1}{2}$ in $C_1(9)$.}
    \label{fig3}
\end{figure}
\end{itemize}
From the above discussion, we have that $x$ has a maximal rotation order of $p-1$ if $x$ is hyperbolic, $p+1$ if it is elliptic, or $p,2p$ if it is parabolic. If any of these cases applies to $x$, we say $x$ is maximal hyperbolic/elliptic/parabolic respectively. A triple is maximal (hyperbolic/elliptic/parabolic) if one if its coordinates is with respect to its corresponding type (again remembering that $x=3x_1$).

Note that the rotation order of a parabolic $x$ is either $p$ or $2p$. We consider both of these elements to be maximal parabolic. If a triple $X$ contains either element, we can connect $X$ to a triple containing any coordinate.

\subsection{The End Game}
We are now ready to delve into the BGS algorithm. In this section we aim to show that any element of order $p^{1/2+\delta}$ for $\delta>0$ can be connected to a triple with a coordinate that is maximal with respect to its type. Later on we will show that every element in $\hat{G}_p$ can be connected to a triple of maximal order, and that such maximal triples themselves can be connected, implying the connectedness of $\hat{G}_p$.

\begin{prop}[Proposition 7 of \cite{BGS1}]\label{endgame}Let $X$ be a triple with rotation order at least $p^{1/2+\delta}$ for $\delta>0$ fixed. Then $X$ is connected to a maximal triple $Y$.\end{prop}
As defined above, the triple $X$ can be classified as hyperbolic, elliptic, or parabolic.

The parabolic case is trivial. As we discussed above, we can connect a parabolic triple to an arbitrary coordinate.

Let us first suppose the element is hyperbolic. Then applying (\ref{hyperbolicmap}) to $X$ gives elements of the form
$$(x_1,\alpha_1t+\alpha_2t^{-1},\alpha_3t+\alpha_4t^{-1})$$
Here $\alpha_i\in\F_p^*$ and $t\in H$, where $H$ is some cyclic subgroup of $\F_p^*$. If we want to connect $X$ to a maximal triple by iteratively applying (\ref{hyperbolicmap}), we would like the second coordinate to eventually take the form $\rho+\rho^{-1}$, where $\rho$ is a primitive root of $F_p^*$. The latter is exactly the form of a maximal hyperbolic element.

So let $P(H)$ denote the number of solutions to
\begin{equation}\label{hyperbolicsolutions}\alpha_1t+\alpha_2t^{-1}=\rho+\rho^{-1}\end{equation}
where $\rho$ is a primitive root of $\F_p^*$.

On the other hand, let $K$ be an arbitrary subgroup of $\F_p^*$. Now define $P(H,K)$ to be the number of solutions to (\ref{hyperbolicsolutions}) where we require $\rho\in K$ instead of $\rho$ being a primitive root.

The subgroups $H$ and $K$ are determined by their indices in $\F_p^*$; set $d_K=(p-1)/|K|$ and $d_H=(p-1)/|H|$.

Now suppose $(t,y)$ is a solution to
\begin{equation}\label{powerhyp}\alpha_1t^{d_H}+\alpha_2t^{-d_H}=y^{d_K}+y^{-d_K}\end{equation}
Then the map $(t,y)\mapsto(t^{d_H},y^{d_K})$ sends solutions of (\ref{powerhyp}) to solutions of (\ref{hyperbolicsolutions}); this map is $d_Hd_K$ to 1.

Thus if $N(\alpha_1,\alpha_2)$ is the number of solutions to (\ref{powerhyp}), then
$$P(H,K)=\frac{N(\alpha_1,\alpha_2)}{d_Hd_K}$$
As shown by Lemma 8 of \cite{BGS1}, the curve
$$\alpha_1t^{d_H}+\alpha_2t^{-d_H}-y^{d_K}-y^{-d_K}$$
given by (\ref{powerhyp}) is absolutely irreducible with genus $O(d_Hd_K)$. Thus applying the Hasse-Weil bound for irreducible curves gives
$$N(\alpha_1,\alpha_2)=p+O(d_Hd_K\sqrt{p})$$
which in turn gives
\begin{equation}\label{weiluse}P(H,K)=\frac{p}{d_Hd_K}+O(\sqrt{p})\end{equation}
We now want to express $P(H)$ in terms of $P(H,K)$. We use inclusion/exclusion on $K$ to eventually find all primitive roots. Let $p_i$ be the distinct prime factors of $p-1$. Also let $K_d$ be the subgroup of $\F_p^*$ of index $d$, e.g. $K_1=\F_p^*$ and $K_{p-1}=\{1\}$. Then we have:
\begin{align}P(H)&=P(H,K_1)-\sum_iP(H,K_{p_i})+\sum_{i,j}P(H,K_{p_ip_j})-\cdots\nonumber\\
&=\sum_{d\mid p-1}\mu(d)P(H,K_d)\label{mobius}
\end{align}
Plugging \eqref{weiluse} into \eqref{mobius} gives
\begin{align*}
P(H)&=\sum_{d\mid p-1}\mu(d)\of{\frac{|H|}{d}+O(\sqrt{p})}\\
&=\of{|H|\sum_{d\mid p-1}\frac{\mu(d)}{d}}+O(p^{1/2+\epsilon})\\
&=\of{|H|\frac{\phi(p-1)}{p-1}}+O(p^{1/2+\epsilon})\\
&\ge|H|(p-1)^{-\epsilon}+O(p^{1/2+\epsilon})
\end{align*}
We assumed our initial triple $X$ had order $\ge p^{1/2+\delta}$, i.e. $|H|\ge p^{1/2+\delta}$. Thus $P(H)>1$ and so there exists at least one solution to equation (\ref{hyperbolicsolutions}). This implies that the orbit of this rotation contains a maximal triple and the hyperbolic case is handled.

The elliptic case is covered in detail in Section 3 of \cite{BGS1}. However, as the technical details of their argument are not needed for the paper at hand, we omit them and move on to showing the collection of maximal elements is connected.


\subsubsection{Connectedness of the Cage}

The vertices in $\hat{G}_p$ corresponding to triples of maximal order form a connected component \cite{BGS1}. Consider $C_j(\alpha)\cap C_k(\beta)$ with $j\ne k$, and without loss of generality let $j=1,k=2$. Also suppose $\alpha,\beta\ne0,\pm2/3$. Going forward we sometimes denote $C_1(\alpha)$ as $(\alpha,?,?)$ and $C_1(\alpha)\cap C_2(\beta)$ as $(\alpha,\beta,?)$.

Then
$$|C_1(\alpha)\cap C_2(\beta)|=|(\alpha,\beta,?)|=0,1,2$$
In particular, the intersection consists of all $\gamma$ such that $\alpha^2+\beta^2+\gamma^2-3\alpha\beta\gamma=0$, which has a solution in $\gamma$ if
$$\leg{9\alpha^2\beta^2-4(\alpha^2+\beta^2)}{p}\ge0$$
In particular
$$|(\alpha,\beta,?)|=1+\leg{9\alpha^2\beta^2-4(\alpha^2+\beta^2)}{p}\ge0$$
So consider the \emph{incidence graph} $I(p)$ of $\hat{G}_p$. The vertices of $I(p)$ are $C_j(\alpha)$ and the number of edges between $C_j(\alpha)$ and $C_k(\alpha)$ is $|(\alpha,\beta,?)|$.

\begin{prop}[Proposition 6 of \cite{BGS1}]\label{incprop}
For $p>10$, the incidence graph is connected and in fact has diameter 2.
\end{prop}
\begin{proof}
We want to connect $C_1(\alpha)$ and $C_2(\beta)$. Thus we want to find $\gamma$ such that both $(\alpha,?,\gamma)$ and $(?,\beta,\gamma)$ are nonempty. So suppose there is a point $(\alpha,l,\gamma)$; solve the quadratic in the second coordinate to see that we must have
$$9\alpha^2\gamma^2-4\alpha^2-4\gamma^2=\lambda^2$$
for some $\lambda$. Similarly we must have that
$$9\beta^2\gamma^2-4\beta^2-4\gamma^2=\mu^2$$
for some $\mu$. Rearrange the two equations into the system
\begin{equation}\label{incidence}\begin{cases}
(9\alpha^2-4)\gamma^2-\lambda^2=4\alpha^2\\
(9\beta^2-4)\gamma^2-\mu^2=4\beta^2
\end{cases}\end{equation}
If $\alpha^2=\beta^2$, then we just take $\lambda=\mu$, and we can reduce \eqref{incidence} to one equation and find an explicit value for $\gamma$. Otherwise \eqref{incidence} is an irreducible curve for which we know a solution in $\gamma$ exists for $p>10$. So the diameter of the incidence graph is at most 2. But of course $C_1(\alpha)$ is not connected to $C_1(\beta)$ if $\alpha\ne\beta$. Thus the diameter is precisely 2.
\end{proof}
Define the \emph{cage} to be the subset of maximal triples. We claim the cage is connected, i.e. path-connected.

Suppose $X$ is a maximal triple with maximal coordinate $\alpha$, say $X=(\alpha,?,?)$. Suppose $Y$ is a maximal triple with maximal coordinate $\beta$, say $Y=(?,\beta,?)$. By Proposition \ref{incprop}, we know there exists a $\gamma$ such that both $(\alpha,?,\gamma)$ and $(?,\beta,\gamma)$ are nonempty. However, we need $\gamma$ to have maximal order:
\begin{center}\begin{tikzcd}[column sep=large]
(\alpha,?,?)\arrow[dash,r,"\alpha\text{ maximal}"]&(\alpha,?,\gamma)\arrow[dash,r,"\gamma\text{ maximal}"]&(?,\beta,\gamma)\arrow[dash,r,"\beta\text{ maximal}"]&(?,\beta,?)
\end{tikzcd}\end{center}
The paper of Bourgain, Gamburd, and Sarnak \cite{BGS1} finishes the proof to guarantee the existence of such a maximal $\gamma$. Thus the cage is connected. We now illuminate this approach through a concrete example.

\subsubsection{Constructive Example}
Let's now walk through a simple example to show how vertices are connected using the BGS algorithm. Take $p=17$. We have the elements along with their order and type in the following table.
\begin{center}\begin{tabular}{c|c|c}
Element&Order&Type\\
\hline
0&4&parabolic\\
1&18&elliptic\\
2&8&hyperbolic\\
3&16&hyperbolic\\
4&16&hyperbolic\\
5&34&parabolic\\
6&6&elliptic\\
7&18&elliptic\\
8&9&elliptic\\
9&18&elliptic\\
10&9&elliptic\\
11&3&elliptic\\
12&17&parabolic\\
13&16&hyperbolic\\
14&16&hyperbolic\\
15&8&hyperbolic\\
16&9&elliptic\\
\end{tabular}\end{center}
Consider the triple $X=(15,0,8)\in G_{17}$. This triple is not maximal, but it does have order $>p^{1/2+\delta}$. By Lemma \ref{endgame}, we should be able to connect $X$ to the cage through rotations of its maximal element.

Since the coordinate 8 has the highest order, we consider $\rot_8$ applied to $X$:
$$(15,0,8)\mapsto(0,2,8)\mapsto(2,14,8)\mapsto(14,11,8)\mapsto(11,12,8)\mapsto(12,5,8)\mapsto$$
$$(5,6,8)\mapsto(6,3,8)\mapsto(3,15,8)\mapsto(15,0,8)$$
which is just a shuffle of the coordinates
$$15-0-2-14-11-12-5-6-3-15$$
for which 14 and 3 are maximal hyperbolic, and 12 and 5 are maximal parabolic. Thus we can connect $X$ to the cage in a number of ways.

A visual representation of this rotation within the plane $C_3(8)$ is given in Figure \ref{rot8}.

\begin{figure}[h!]
\centering
\includegraphics[scale=.22]{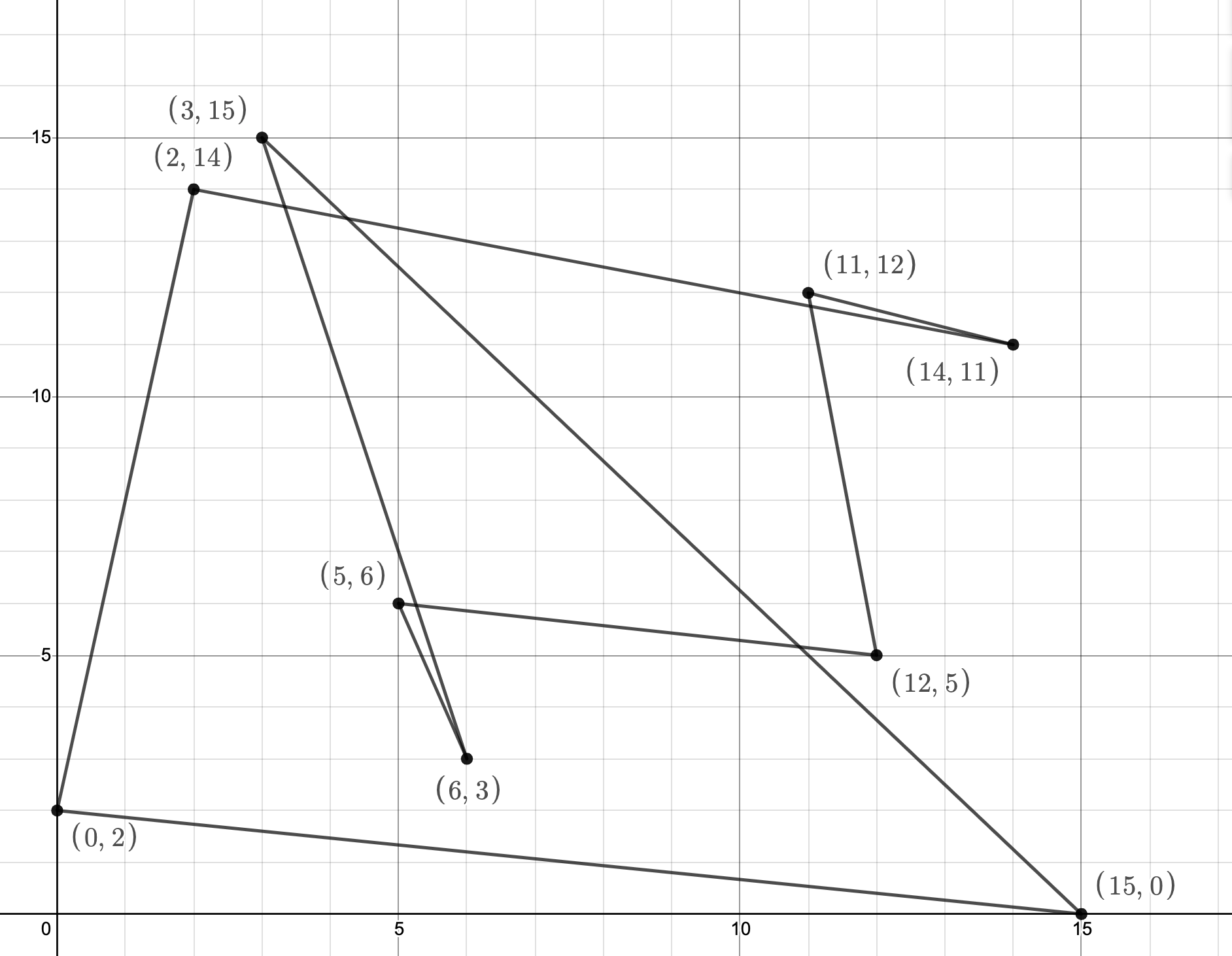}
\caption{$\rot_8$ applied to $(15,0,8)$ in the plane $C_3(8)$.}
\label{rot8}
\end{figure}
\subsection{The Middle Game and The Opening}
In this section we aim to show that any triple $X$ of small order can be  connected to the cage in a finite number of moves. By small order, we mean triples $X$ whose order is $p^\varepsilon$ (which we refer to as \textit{The Middle Game}) or those whose order is less than $p<c$ for some constant $c$, i.e. points whose orders are uniformly bounded independent of $p$ (which we refer to as \textit{The Opening}).

We first handle the Middle Game in detail, and then outline how the Opening comes into play. In particular, we connect a triple of order $p^\varepsilon$ to the cage by showing that one can connect it to a triple whose order is strictly greater than that of the original triple, and then iterate the process until we have a triple of order $p^{1/2+\delta}$ and we are in the End Game. This is done via the following procedure:

Define the \emph{maximal orbit} $M_X$ of a triple $X$ as the orbit corresponding to the rotation of the maximal coordinate of $X$. So if $X$ is a triple with order $l$, then $|M_X|=l$. Any orbit occurs with respect to either the first, second, or third coordinate; we call this number the index of the orbit.
\begin{enumerate}
\item Let $Y\in M_X$ and $l_Y$ be the order of $Y$. Of course $l,l_Y\mid p^2-1$.
\item If $l_Y>l$, then necessarily the index of $M_Y$ is not equal to the index of $M_X$. Then replace $X$ with $Y$, thereby strictly increasing the order of $X$.
\item Otherwise, $l_Y\le l$. Consider the sum
\begin{equation}\label{middlegamesum}N_{l}=\sum_{l'\le l}\#\{Y \in M_X:l_Y=l'\}\end{equation}
If $N_{l} < l$, then there must be a point $Z\in M_X$ whose order is strictly greater than that of $X$. We can then replace $X$ with $Z$ and repeat this process until we arrive at an element with order at least $p^{1/2+\delta}$, which is a reduction to the endgame. This must happen because the order strictly increases at each step and $p^2-1$ has finitely many divisors.
\end{enumerate}
Now we need to effectively bound $N_l$. As seen in the endgame, every $Y$ (with order $l_Y\mid p^2-1$) in the maximal orbit of $X$ corresponds to a solution of the equation:
\begin{equation}\label{bgs43}
\begin{cases}
h_1+\frac{\sigma}{h_1}=h_2+\frac{1}{h_2}\\
h_1 \in H_1, h_2 \in H_2\\
H_1,H_2 \text{ subgroups of } \F_p^* \text{ or } \F_{p^2}^*\\
\sigma\in\F_p
\end{cases}
\end{equation}
where $|H_1|=l$ and $|H_2|=l_Y$. Also, from equation \eqref{notone}, we have that $\sigma\neq1$. So we see that in fact $N_l$ denotes precisely the number of solutions to \eqref{bgs43}, so equivalently we want an upper bound on the number of solutions to \eqref{bgs43}. The following bound is derived in \cite{BGS1}, based off previous work of Bourgain (Proposition 2 in \cite{B10}).
    
\begin{prop}[Proposition 10 of \cite{BGS1}]
Given $\delta > 0$ there is $\tau < 1$ and $C_\tau$ depending on $\delta$ such that if
$p^\delta < |H_1| < p^{1-\delta}$ then the number of solutions to \eqref{bgs43} is at most $C_\tau |H_1|^\tau$.
\end{prop}
From this proposition, we simply deduce
$$N_l\le C_\tau|H_1|^\tau=C_\tau l^\tau$$
which provides a necessary upper bound to the number of solutions of \eqref{bgs43}, as desired.

Thus any triple of order at least $p^\epsilon$ can be connected to the cage, and so all triples of order at least $p^\epsilon$ are connected. This algorithm is essential for our cryptographic constructions, and provides the backbone to the first step in connecting two triples $X$ and $Y$ as discussed in Section \ref{cryptheur}.

Next we consider the part of the BGS algorithm that is called ``the Opening," in \cite{BGS1}: that is, the rest of the points in $\hat{G}_p$ whose order is less than $p<c$ for some constant $c$, that is points whose orders are uniformly bounded.

In the Opening section of \cite{BGS1}, Bourgain, Gamburd, and Sarnak prove that one can connect triples with uniformly bounded orders to the cage to conclude that the Markoff graph mod $p$ is connected; however their methods are non-constructive. To go about this, they look at the characteristic 0 case and show that there are no finite $\Gamma$-orbits. As this method is not needed in our cryptographic analysis of $\hat{G}_p$, we omit the technicalities and direct the interested reader to Section 5 of \cite{BGS1} for a comprehensive analysis of the Opening. 

The proof of Theorem \ref{BGSmain} presented throughout Section \ref{BGSattack} provides us with an algorithmic approach to finding paths in $\hat{G}_p$, thus establishing connectivity of $\hat{G}_p$. This method need not be optimal but the cryptographic analysis in Section \ref{cryptheur} elucidates the strength of the cryptosystem against the BGS-style attack. We now look at another possible avenue of path-finding based off lifting solutions to $\Z$ and exploiting the structure of the Markoff tree. 

\section{Attack by Lifting}\label{Attack2}
The main observation behind our plan of attack is the following lemma.

\begin{lem}
Let $(x_1,x_2,x_3)$ be a Markoff triple in $\mathbb{Z}^3$ whose $i$-th coordinate $x_i$ is maximal, and $x_i > 1$. Then applying $R_i$ to the triple decreases the size of the $i$-th entry. Formally, suppose $|x_i|\geq |x_k|$ for all $1\leq k\leq3$ in the Markoff triple $(x_1,x_2,x_3)$.  Let $(x_{1,i},x_{2,i},x_{3,i})$ be the triple obtained from applying the $i$-th involution $R_i$ to the triple: $$(x_{1,i},x_{2,i},x_{3,i}):=R_i(x_1,x_2,x_3).$$  
Then $|x_{i,i}|<|x_i|$.
\end{lem}

\begin{proof}
Let $x_j,x_k$ be the other two coordinates of the triple $(x_1,x_2,x_3)$ besides $x_i$. Note that, since $|x_i|>1$, it is impossible for $(x_i,x_j,x_k)$ to satisfy (\ref{markoffeq}) if $|x_i|=|x_j|=|x_k|$.  In fact, in this case we have that $|x_i|$ must be strictly larger than $|x_j|$ and $|x_k|$ in order for (\ref{markoffeq}) to be true.  Suppose further without loss of generality that $|x_j|\leq |x_k|$.

We have $x_{i,i}=3x_jx_k-x_i$.  If $x_i>0$, then $x_jx_k>0$ in order for (\ref{markoffeq}) to be satisfied, and, again by (\ref{markoffeq}) we have $$3x_jx_k=(x_1^2+x_2^2+x_3^2)/x_i>x_i,$$ so that $$|x_{i,i}|=|3x_jx_k-x_i|=3x_jx_k-x_i.$$
Our goal is hence to show that $2x_i-3x_jx_k>0$.  We have by (\ref{markoffeq}) that
$$2x_i-3x_jx_k=2x_i-\frac{x_i^2+x_j^2+x_k^2}{x_i}=\frac{x_i^2-x_j^2-x_k^2}{x_i},$$
which, given that $x_i>0$, is positive if and only if the numerator is positive.  Rewrite the numerator as
$$x_i^2+x_j^2+x_k^2-(2x_j^2+2x_k^2)$$
and compare with the left side of (\ref{markoffeq}).  We claim that $2x_j^2+2x_k^2<3x_ix_jx_k$, which would imply that the numerator above is positive as desired.

It remains to prove our claim.  Given that $x_i>|x_j|$ and $|x_k|\geq|x_j|$, we have
$$3x_ix_jx_k>x_ix_jx_j+2|x_k|x_jx_k\geq 2x_j^2+2x_k^2$$
as desired where the last inequality is true since $x_i\geq 2$ and $|x_j|\geq 1$.  So, if $x_i>0$ we are done.

If $x_i<0$ the argument is nearly identical.  We would have that $x_{i,i}<0$ in that case, and so our goal would be to show that $-3x_jx_k+x_i<-x_i,$ or that $2x_i-3x_jx_k<0$.  Given that $(-x_i,-x_j,x_k)$ is a triple satisfying the properties in the first case above where $x_i>0$, the argument above shows that $-2x_i+3x_jx_k>0$, which is exactly what we need. 
\end{proof}

This lemma gives a very straightforward way of finding a path from any triple $(x_1,x_2,x_3)$ in the tree to the triple that is the ``origin," or $(1,1,1)$ in absolute value, which in turn gives a simple way of finding a path between any two vertices in the tree.  Thus if triples can be efficiently lifted from the graph $G_p$ to the tree, this algorithm gives a path-finding attack on the graph.  The algorithm is as follows.  Start with $W=I$, the identity.
\begin{enumerate}
\item If $(|x_1|,|x_2|,|x_3|)=(1,1,1)$ then we are done, and $W$ is the word that describes the path from $(x_1,x_2,x_3)$ to the origin. If not, determine $i$ such that the $i$-th coordinate of $(|x_1|,|x_2|,|x_3|)$ is largest. Go to step 2.
\item Replace $(x_1,x_2,x_3)$ with $R_i((x_1,x_2,x_3))$, replace $W$ with $WR_i$, and go to step 1.
\end{enumerate}
By the lemma, this algorithm will continuously decrease every largest coordinate in absolute value until each coordinate is $1$ in absolute value.  For example, for the triple $(29,-169,-14701)$ it gives
$$(29,-169,-14701)\xrightarrow{R_3}(29,  -169,    -2)\xrightarrow{R_2}(29,    -5,    -2)\xrightarrow{R_1}(1,    -5,    -2)$$
$$\xrightarrow{R_2}(1,    -1,    -2)\xrightarrow{R_3}(1,-1,-1).$$
Coming back to our problem of finding paths between two points in the graph $G_p$,, if our attacker is able to take a triple $(x_1',x_2',x_3')$ which satisfies the Markoff equation modulo $p$ and lift it to a solution $(x_1,x_2,x_3)$ to the Markoff equation in $\mathbb Z$, then she need only run the algorithm above to find a path from $(x_1,x_2,x_3)$ to the origin in which every coordinate is $1$ in absolute value in order to find a path from $(x_1',x_2',x_3')$ to the origin in $G_p$ (it is the path corresponding to the same word as the one she will obtain from the above algorithm).  

However, so far it appears that finding a Markoff triple that reduces to $(x_1',x_2',x_3')$ modulo $p$ is difficult for most candidate $(x_1',x_2',x_3')$'s. The reason for this is that, according to \cite{Z}, the number of Markoff triples in which the largest coordinate is at most $T$ is asymptotic to $C(\log T)^2$ for some constant $C$, while the number of vertices in $G_p\sim p^2$. So in order to have a chance of covering all possible mod-$p$ Markoff triples coming from $G_p$ by Markoff triples over $\mathbb Z$, one must consider all those triples less than $T$ where
$$C(\log T)^2\geq p^2,$$
or, in other words, where $T$ is of size roughly $e^p$. More likely, $T$ will have to be much larger than that, since it is not at all true that all Markoff numbers less than $T$ reduce to a different triple modulo $p$. Even with this estimate of $e^p$, however, one sees that the lifts will probably be very large (since $p$ itself will be taken to be large), and certainly no straightforward search for a lift in $\mathbb Z$ will be computationally feasible.

Constructing a collision attack from lifting is almost equivalent to path finding. If one has a method of efficiently finding lifts to $\Z$, two lifts of the same triple could result in two distinct paths between triples. Unless the two paths in $\Z$ overlap nontrivially, we would have a collision starting with $(0,0,0)$.

\section{Other Possible Attacks and Future Avenues for Research}\label{otherattacks}
We note that the BGS algorithm can be slightly modified to search for collision resistance. Currently, the steps of the middle game are deterministic in connecting a triple to the cage; the rotations are always done on the maximal coordinate. For a collision attack, we would search for two distinct paths between points. So instead of always choosing the maximal coordinate, we can randomly choose coordinates instead (not necessarily uniformly). If we eventually arrive at the cage, then we have found another distinct path, since the cage is connected. Of course, there is no proof, other than empiricism, that any method other than choosing the maximal coordinate will succeed in a similar way.

Many potential attacks involve finding small cycles on $G_p$ or $\hat{G}_p$, e.g. some adaptation of the Pollard rho algorithm. There are a number of reasons we believe such a study is unfruitful. A Pollard-style attempt would look for cycles by repeatedly applying a single involution. The construction of $G_p$ means that such short cycles occur with vanishingly little frequency, as discussed in the Opening. In any case, such discrete logarithm attacks must involve at least $\Omega(\sqrt{p})$ group operations \cite{S}, which is not a significant improvement.

Nonetheless, it will certainly be important to understand better the distribution of cycle lengths in a graph $G_p$ or $\hat{G}_p$. While it is known that small cycles in $\hat{G}_p$ exist, it is not known how common they are, and how likely one is to run into one in practice. Even less is known about the cycles in the graph $G_p$. This is a problem the authors hope to explore in a future paper.

In addition, it would be helpful to have a better picture of the size of an average lift of a Markoff triple mod $p$ to one over $\mathbb  Z$, so that we can further understand the potential for success of the lifting attack described in Section~\ref{Attack2}. This is currently being studied by the first-named author together with co-authors E.~Bellah,  S.~Kim, D.~Schindler, J.~Sivaraman, and L.~Ye.

\end{document}